\RequirePackage[l2tabu,orthodox]{nag}
\documentclass
[11pt,letterpaper]
{article} 

\usepackage[notes=true,later=false,camera=false]{dtrt}
\usepackage[utf8]{inputenc}
\usepackage{ stmaryrd }
\usepackage{xspace,enumerate}
\usepackage[T1]{fontenc}
\usepackage[full]{textcomp}
\usepackage[american]{babel}
\usepackage{mathtools}

\renewcommand{\hat}[1]{\widehat{#1}}
\usepackage{amsthm}
\usepackage{thmtools}
\usepackage[capitalise,nameinlink]{cleveref}

\usepackage{empheq}

\definecolor{citeblue}{HTML}{0055cc}
\hypersetup{
colorlinks=true,
urlcolor=citeblue,
linkcolor=NavyBlue,
citecolor=PineGreen,
linktocpage=true,
}
\renewcommand*\backref[1]{\ifx#1\relax \else (pg. #1) \fi}

\renewcommand{\tilde}{\widetilde}

\usepackage{paralist}
\usepackage{turnstile}
\usepackage[framemethod=TikZ]{mdframed}
\mdfsetup{frametitlealignment=\center}
\usepackage{tikz}
\usepackage{caption}
\DeclareCaptionType{Algorithm}
\usepackage{newfloat}

\newtheorem{theorem}{Theorem}[section]
\newtheorem{lemma}[theorem]{Lemma}
\newtheorem*{lemma*}{Lemma}

\newtheorem{proposition}[theorem]{Proposition}
\newtheorem{fact}[theorem]{Fact}

\theoremstyle{definition}
\newtheorem{definition}[theorem]{Definition}
\newtheorem*{definition*}{Definition}
\newtheorem{remark}[theorem]{Remark}

\usepackage[linesnumbered,ruled]{algorithm2e}
\crefname{lemma}{Lemma}{Lemmas}
\crefname{fact}{Fact}{Facts}
\crefname{theorem}{Theorem}{Theorems}
\crefname{mtheorem}{Theorem}{Theorems}
\crefname{itheorem}{Theorem}{Theorems}
\crefname{corollary}{Corollary}{Corollaries}
\crefname{claim}{Claim}{Claims}
\crefname{example}{Example}{Examples}
\crefname{algorithm}{Algorithm}{Algorithms}
\crefname{problem}{Problem}{Problems}
\crefname{definition}{Definition}{Definitions}
\crefname{equation}{Eq.}{Eq.}
\crefname{strategy}{Strategy}{Strategies}
\crefname{observation}{Observation}{Observations}

\Crefname{algocf}{Algorithm}{Algorithms}

\usepackage[
letterpaper,
top=1.2in,
bottom=1.2in,
left=1in,
right=1in]{geometry}
\usepackage{newpxtext} %
\usepackage{textcomp} %
\usepackage[scr=rsfso]{mathalfa}
\usepackage{bm} 

\usepackage{microtype}

\usepackage{footnotebackref}

\renewcommand{\bar}{\overline}

\newcommand{\supp}{\operatorname{supp}}
\newcommand{\AND}{\operatorname{AND}}

\newcommand{\D}{\operatorname{D}}
\newcommand{\U}{\operatorname{U}}
\newcommand{\Sn}{\mathfrak{S}_n}

\allowdisplaybreaks
\newcommand{\FormatAuthor}[3]{
\begin{tabular}{c}
#1 \\ {\small\texttt{#2}} \\ {\small #3}
\end{tabular}
}

\newcommand{\R}{{\mathbb R}}

\newcommand{\eps}{\varepsilon}

\newcommand{\E}{{\mathbb E}}

\newcommand{\1}{\mathbf{1}}

\newcommand{\C}{\mathbb C}
\newcommand{\Bits}{\{0,1\}}
\newcommand{\zo}{\Bits}
\newcommand{\bra}[1]{\langle #1\rvert}
\newcommand{\ket}[1]{\lvert #1 \rangle}
\newcommand{\braket}{\ket}

\newcommand{\cH}{\mathcal H}

\newcommand{\mper}{\,.}
\newcommand{\mcom}{\,,}

\newcommand{\Id}{\operatorname{Id}}

\newcommand{\seq}{\subseteq}

\newcommand{\polylog}{\operatorname{polylog}}

\newcommand{\cK}{\mathcal K}

\newcommand{\cM}{\mathcal M}

\newcommand{\cL}{\mathcal L}

\renewcommand{\geq}{\geqslant}
\renewcommand{\ge}{\geqslant}
\renewcommand{\leq}{\leqslant}
\renewcommand{\le}{\leqslant}
\renewcommand{\preceq}{\preccurlyeq}
\renewcommand{\succeq}{\succcurlyeq}
\renewcommand{\epsilon}{\varepsilon}

\newcommand{\ignore}[1]{}

\newcommand{\id}{\operatorname{id}}
\newcommand{\dmix}{\operatorname{dmix}}
\newcommand{\cAK}{c_{\mathsf{AK}}}
\newcommand{\Vol}{\mathrm{Vol}}
\newcommand{\QC}{\mathrm{QC}}

\begin{document}

\title{Quantum Cut Sparsifiers}
\author{
\begin{tabular}{cc}
\FormatAuthor{Arpon Basu}{arpon.basu@princeton.edu}{Princeton University} &
\FormatAuthor{Joshua Brakensiek\thanks{Supported by NSF Awards CCF-2211972 and DMS-2503280, ONR Grant N00014-24-1-2491, and a Simons Investigator award of Venkatesan Guruswami.}}{josh.brakensiek@berkeley.edu}{UC Berkeley} \\
 & \\
\FormatAuthor{Pravesh K. Kothari}{kothari@cs.princeton.edu}{Princeton University} &
\FormatAuthor{Aaron Putterman\thanks{Supported in part by the Simons Investigator Awards of Madhu Sudan and Salil Vadhan and AFOSR award FA9550-25-1-0112.}}{aputterman@g.harvard.edu}{Harvard University}
\end{tabular}
}
\maketitle

\begin{abstract}
In this paper, we continue a line of research initiated by Basu, Brakensiek, and Putterman [2026] studying the sparsifiability of Hamiltonians. We focus particularly on the sparsifiability of the widely-studied Quantum Cut (QC) Hamiltonians. Our main result is that in an $n$-qubit system, any $n$-qubit QC Hamiltonian can be sparsified to $\widetilde{O}(n /\epsilon^2)$ many terms while preserving the energy of every state up to a factor of $1 \pm \epsilon$. Our result can be interpreted as giving an importance sampling scheme for the edges of an arbitrary graph $G$ such that the \emph{Kikuchi} graph at level $\ell$ of the sampled graph is a spectral approximation to the Kikuchi graph of $G$. Importantly, the \emph{same} sampling scheme works simultaneously for all $\ell$. 

The natural approach of leverage score sampling, analyzed via matrix concentration inequalities, yields a polynomially worse bound in our setting because the underlying matrices have dimension $\sim 2^n$. Instead, our approach relies on decomposing the action of these matrices into invariant subspaces. Then, by using an operator-valued inequality of Alon and Kozma [Ann.~Henri Poincar\'e, 2020], itself building on an \emph{octopus inequality} of Caputo, Liggett, and Richthammer [J.~AMS, 2010], we extend our sparsification technique to all expander graphs. We then invoke expander decomposition to extend our sparsifier to all graphs.
\end{abstract}

\pagenumbering{gobble}

\clearpage

\pagebreak

\tableofcontents

\pagebreak

\pagenumbering{arabic}
\section{Introduction}

Sparsification is a fundamental algorithmic question, asking when large complex objects admit much smaller representations which still preserve key properties. In many settings it turns out that efficient, significant sparsification is possible, thereby dramatically reducing the size of objects. This paradigm was introduced in the seminal work of Bencz\'ur and Karger \cite{BenczurK96} who showed that arbitrary graphs $G = (V, E)$ admit $\eps$ \emph{cut-sparsifiers} with only $O\left (\frac{|V|\log|V|}{\eps^2} \right)$ many edges. These cut-sparsifiers preserve cut-sizes in the original graph $G$; i.e., for any set $S \subseteq V$, the weight of the cut $S$ in the cut-sparsifier is within a $(1 \pm \eps)$ factor of the weight of the same cut in $G$. Naturally, these cut-sparsifiers have become an important algorithmic tool, as they allow for significant size-reduction in the input graphs to many algorithms. 

Since the initial conception of this notion in \cite{BenczurK96}, sparsification has developed into wide-ranging field, with breakthrough works studying \emph{spectral sparsification} of graphs \cite{SpielmanT11, SpielmanS11, BatsonSS14}, sparsification of hypergraphs \cite{ChenKN20, JambulapatiLS23, Lee23}, codes \cite{KhannaPS24, BrakensiekG25}, and CSPs \cite{KoganK15, FiltserK17, KhannaPS24, BrakensiekG25}. In this paper, we continue a more recent line of research initiated by Basu, Brakensiek, and Putterman \cite{BasuBP26} which generalizes this work on sparsification to the setting of quantum \emph{Hamiltonians}. In particular, in this work, we study the sparsifiability of one of the most-widely studied classes of Hamiltonian, namely the \emph{Quantum Cut} (QC) Hamiltonians\footnote{Note that often these are referred to as Quantum MaxCut Hamiltonians.}. Formally, we consider the $2$-qubit Hamiltonian 
\[
M_{\QC} := (\ket{01}-\ket{10})(\bra{01}-\bra{10}) = \frac{1}{2}\left(\Id\otimes\Id - X\otimes X - Y\otimes Y - Z\otimes Z\right)\in\C^{2^2\times 2^2}, 
\]
where $$ \quad \Id = \begin{pmatrix} 1 & 0 \\ 0 & 1 \end{pmatrix}, \quad X = \begin{pmatrix} 0 & 1 \\ 1 & 0 \end{pmatrix}, \quad Y = \begin{pmatrix} 0 & -\sqrt{-1} \\ \sqrt{-1} & 0 \end{pmatrix}, \quad Z = \begin{pmatrix} 1 & 0 \\ 0 & -1 \end{pmatrix} $$
are the usual Pauli matrices in $\C^{2\times 2}$. 

For a weighted graph $G = (V, E, w)$ with vertex set $V$, edge set $E \subseteq \binom{V}{2}$, and edge weights $w : E \to \R_{\ge 0}$, we define the $2$-local Hamiltonian
\begin{equation}
\label{eq:qcpredicate}
    \cL_{G} := \sum_{e \in E} w_e \cL_e,
\end{equation}
where $\cL_e:= M_{\QC}|_{e} \otimes \Id_{V \setminus e}$ acts on the qubits in $e$. As mentioned above, these QC Hamiltonians are very widely studied in the literature, with works such as \cite{brandao2013product, GharibianP19, anshu2020beyond, parekh2021application, parekh2022optimal, king2023improved, huber2024second, lee2024improved, apte2025improved, gribling2025improved} focusing on understanding the ability of efficient algorithms to approximate the \emph{maximum energy} achievable (AKA, the Max-CUT) of a QC instance (i.e, $\max_{\braket{\psi}: \Vert \braket{\psi}\Vert_2 = 1} \langle \psi | \cL_G |\psi \rangle$). Other works have instead considered the \emph{hardness} of this problem, with \cite{hwang2023unique, piddock2025quantum} showing conditions under which it is hard to approximate the maximum energy. Even further still, other works such as \cite{KallaugherP22} have studied streaming algorithms for approximating Quantum Max-CUT. 

Motivated by this extensive research into QC Hamiltonians, we now seek to understand the ``sparsifiability'' of QC Hamiltonians. 
More formally, our goal is to construct a \emph{sparsifier} of $\cL_G$. 
\begin{definition}[Quantum Cut Sparsifier]\label{def:qcs}
Given a graph $G = (V, E, w)$ and some $\eps > 0$, a graph $\widetilde{G} = (V, \widetilde{E}, \widetilde{w})$ is an $\eps$-Quantum Cut ($\eps$-QC) sparsifier of $G$ if for all states $\ket{\psi} \in \C^{\{0,1\}^V}$, we have that 
\begin{align}
    (1 - \eps) \bra{\psi}\cL_{G}\ket{\psi} \preceq \bra{\psi}\cL_{\widetilde{G}}\ket{\psi} \preceq (1 + \eps) \bra{\psi}\cL_{G} \ket{\psi}.\label{eq:QCS}
\end{align}
\end{definition}
In other words,
\[
    (1 - \eps) \cL_{G} \preceq \cL_{\widetilde{G}} \preceq (1 + \eps) \cL_{G},
\]
where $A \preceq B$ is notation for $B - A$ being positive semidefinite. Importantly, this preserves the energy of the QC Hamiltonian \emph{on every single} state $\braket{\psi}$, and thus ensures a ``full-spectrum'' simulation of the original Hamiltonian \cite{zhou2021strongly, csahinouglu2021hamiltonian, zlokapa2024hamiltonian}.
Note that the motivation for sparsification is well-founded: building sparsifiers would allow for a significant reduction of the space complexity of representing any such Hamiltonian, as it allows for the storage of a \emph{sparsifier} in place of the original Hamiltonian. Not only this,  but sparsification is also an important \emph{algorithmic} primitive. To see this, suppose that one has a QC Hamiltonian with $m$ local terms (edges). Naively, one might run an algorithm on this Hamiltonian which has a runtime that scales as $\mathrm{poly}(m)$. Sparsification instead allows for an efficient ``preprocessing phase'' which (ideally) \emph{reduces} the number of terms to a much smaller $m' \ll m$ \emph{without changing} the spectrum of the Hamiltonian. From here, one can instead run this same algorithm \emph{only} on the sparsifier, yielding a runtime which now scales as $\mathrm{poly}(m')$ as opposed to $\mathrm{poly}(m)$. Indeed, it was this key observation which motivated the first breakthrough works on sparsification \cite{BenczurK96}. Motivated by the above, the goal of this work is to answer the following basic question for QC Hamiltonians:

\begin{center}
\emph{What is the smallest possible $\eps$-QC sparsifier for a given input graph?}
\end{center}

If we let $n = |V|$ be the number of vertices of our input graph $G = (V, E, w)$, then \emph{a priori} the answer could land anywhere between $\Omega(n)$ and $O(n^2)$. The work of Aharonov and Zhou~\cite{AZ19} gave evidence that the answer is $\Theta(n^2)$ (i.e., essentially no sparsification is possible) by constructing an explicit $2$-qubit Hamiltonian ($M_{\mathrm{AZ}} := \ket{11}\bra{11}$) for which sparsification (in a sense analogous to \cref{def:qcs}) is provably impossible. In fact, this result by \cite{AZ19} was cited by Kallaugher and Parekh~\cite{KallaugherP22} as a strong barrier to Quantum Cut sparsification.

However, a recent paper by Basu, Brakensiek, and Putterman~\cite{BasuBP26} puts doubts on the consensus opinion that Quantum Cut sparsification is impossible. Inspired by recent advances in the theory of (classical) CSP sparsification (e.g., \cite{KoganK15,FiltserK17,ButtiZ20,BessiereCK20,ChenKN20,ChenJP20,LagerkvistW20,KapralovKTY21,KapralovKTY21a,KhannaPS24,KhannaPS25,BrakensiekG25}), \cite{BasuBP26} introduced the notion of \emph{non-redundancy} of local Hamiltonians which measures how fragile the ground space of the Hamiltonian is to changes in the local terms. In particular, \cite{BasuBP26} shows that the $2$-qubit Hamiltonian $M_{\mathrm{AZ}}$ has $\Omega(n^2)$ non-redundancy in the sense that even in a complete graph deleting a single edge of the graph $G$ changes the ground state of the resulting Hamiltonian. On the other hand, \cite{BasuBP26} show for $M_{\QC}$ that no similar barrier exists--the non-redundancy of Quantum Cut is $O(n)$. As a result, \cite{BasuBP26} explicitly ask if there exist $\eps$-QC sparsifiers with $\widetilde{O}_{\eps}(n)$\footnote{Here $\widetilde{O}_{\eps}$ hides multiplicative factors of the form $f(\eps)(\log n)^{O(1)}$, where $f(\eps)$ is an arbitrary function of $\epsilon$.} edges

\subsection{Our Results}

Our main result is a positive answer to a question left open in \cite{BasuBP26} by showing that near-linear size Quantum Cut sparsifiers do indeed exist. 

\begin{theorem}[Main, Informal]\label{thm:main}
For any graph $G = (V, E, w)$ with $n = |V|$ and any $\eps > 0$, there exists an $\eps$-QC sparsifier $\widetilde{G} = (V, \widetilde{E}, \widetilde{w})$ with $|\widetilde{E}| \le \widetilde{O}(n/\eps^2)$. Furthermore, if $G$ is an unweighted graph, such a sparsifier can be computed in near-linear time (in the number of edges $|E|$) via a \textbf{classical} algorithm.
\end{theorem}

See \cref{thm:KikuchiSparsifierExistence} and \cref{thm:KikuchiSparsifierEfficient} for precise formulations of Theorem~\ref{thm:main}. The result of \cref{thm:main} can be stated in purely combinatorial terms using \emph{Kikuchi graphs}, see \cref{subsec:related-work} for discussion.

\subsection{Quantum Cuts and Kikuchi Graphs}
\label{subsubsec:kikuchiquantumconnection}
The quantum-cut Hamiltonian from \cref{eq:qcpredicate} has a particularly appealing combinatorial interpretation in terms of \emph{Kikuchi graphs}, a.k.a as \emph{token graphs}. 

Towards elucidating the connection, first consider any function $\psi:\zo^{[n]}\to\C$ (which can be interpreted as a vector in $\C^{\zo^{[n]}}$), and fix any $e = \{u, v\}\in\binom{[n]}{2}$. Then, the energy of a QC Hamiltonian can be written as
\[\langle\psi, \cL_e\psi\rangle = \langle\psi, (M_{\text{QC}})_e\otimes\Id_{[n]\setminus e}\psi\rangle = \sum_{z\in\zo^{[n]\setminus e}}\langle\psi_z, (\ket{01}-\ket{10})(\bra{01}-\bra{10})\psi_z\rangle,\]
where $\psi_z:\zo^{e}\to\C$ is the function defined as $\psi_z(y):= \psi(z, y)$, where $(z, y)\in\zo^{[n]\setminus e}\times\zo^e\cong\zo^{[n]}$ is the concatenation of $z$ and $y$. Simplifying further yields
\[\langle\psi, \cL_e\psi\rangle = \sum_{z\in\zo^{[n]\setminus e}}|\psi(z, 01)-\psi(z, 10)|^2\mper\]
The above expression naturally suggests a graph on $\zo^{[n]}$: For every $z\in\zo^{[n]\setminus e}$, connect $(z, 01), (z, 10)\in\zo^{[n]}$. If we call this graph $\cK_e$, then note that $\langle\psi, \cL_e\psi\rangle = \sum_{z\in\zo^{[n]\setminus e}}|\psi(z, 01)-\psi(z, 10)|^2$ is precisely
\[\sum_{\{x_1, x_2\}\text{ is an edge of }\cK_e}|\psi(x_1)-\psi(x_2)|^2,\]
which is exactly equal to the quadratic form of $\psi$ with the Laplacian of the graph $\cK_e$, i.e. $\cL_e$ is the Laplacian of $\cK_e$! The underlying graph $\cK_e$ also makes clear the nomenclature of the predicate $M_{\text{QC}}$ being \emph{quantum cut}. 

Before we formalize this definition, it is convenient to rephrase it slightly: Note that we can naturally identify $\zo^{[n]}$ with the power set $2^{[n]}$. Then the adjacency condition of $\cK_e$ becomes that $S, T\in 2^{[n]}$ are connected if there exists $Z\in 2^{[n]\setminus e}$ such that $S = Z\cup\{u\}, T = Z\cup\{v\}$ (or $S = Z\cup\{v\}, T = Z\cup\{u\}$), i.e. $S, T$ are connected if $|S| = |T|$ and $S\oplus T = \{u, v\} = e$, where $\oplus$ refers to the symmetric difference of sets. We thus formalize the definition of a Kikuchi graph as follows.
\begin{definition}[Kikuchi Graphs]
\label{def:Kikuchidef}
    Let $e = \{u, v\}\in\binom{[n]}{2}$ be an edge. Define the \emph{Kikuchi graph} $\cK_e$ to be a graph with the vertex set $2^{[n]}$, where $S, T\in 2^{[n]}$ are connected if $|S| = |T|$ and $S\oplus T = e$. For a general weighted graph $G = ([n], E, w)$, define the Kikuchi graph $\cK_G:= \bigcup_{e\in E(G)}w_e\cK_e$, i.e. the vertices of $\cK_G$ are elements of $2^{[n]}$, and $S, T\in 2^{[n]}$ are connected if $S\oplus T\in E(G)$, with the weight of $\{S, T\}$ being $w_{S\oplus T}$. 
\end{definition}

By the discussion above, we also know that $\cL_G$ (as defined in \cref{eq:qcpredicate}) is the Laplacian of $\cK_G$. Now, note that since $S, T\in 2^{[n]}$ can be adjacent only if $|S| = |T|$, we have that for any graph $G$, $\cK_G$ is the disjoint union of $n + 1$ graphs on the \emph{slices} $\binom{[n]}{k}:= \{S\in 2^{[n]}: |S| = k\}$ for $0\leq k\leq n$. Fix any $0\leq k\leq n$, and let $\cK_{G, k}:= \cK_G[\binom{[n]}{k}]$ be the restriction of $\cK_G$ to $\binom{[n]}{k}$. 

\begin{definition}[Level $k$ Kikuchi graphs]
    For a general weighted graph $G = ([n], E, w)$, define the \emph{level $k$} Kikuchi graph $\cK_{G, k}$ to be a graph on $\binom{[n]}{k}$, with $S, T\in\binom{[n]}{k}$ being connected if $S\oplus T\in E(G)$, with the weight of $\{S, T\}$ being $w_{S\oplus T}$. 
\end{definition}
Note that $\cK_{G, 1}$ is isomorphic to $G$. For $k\geq 1$, the graph $\cK_{G, k}$, also known as the \emph{$k$-token graph}, naturally arises in problems related to the quantum cut predicate, and understanding the spectra of these Kikuchi graphs provides the best-known approximation algorithms for the quantum max cut problem, and other related predicates such as the ``XY'' predicate \cite{AptePS25,BakshiBKL26}. Kikuchi graphs/token graphs also arise in various contexts throughout classical computer science and combinatorics, for instance, in positivstellensatz lower bounds for the knapsack problem \cite{Grigoriev01}, in tensor PCA \cite{WeinEAM25}, in the proof of the hypergraph Moore bound \cite{GuruswamiKM22,HsiehKM23}, in refuting random CSPs \cite{GuruswamiKM22,HsiehKM23}, in proving lower bounds for locally decodable and correctable codes \cite{AlrabiahGKM23,BasuHKL25,JanzerM25,KothariM23}, and even in proving data structure lower bounds \cite{KortenPI25,GuruswamiLY26}! 

Let us now tie all the discussion above with sparsification, which is the main focus of this paper. The notion of Quantum Cut sparsification that we consider in \cref{def:qcs} is much more general than classical and spectral notions of cut sparsification considered in the literature. In particular, if we only demand that \cref{eq:QCS} holds for states of the form $\ket{\psi} = \sum_{i =1}^n c_i \ket{0}^{i-1}\ket{1}\ket{0}^{n-i}$ with each $c_i \in \{-\frac{1}{\sqrt{n}}, \frac{1}{\sqrt{n}}\}$, then we recover the original notion of cut sparsifiers due to Bencz\'{u}r and Karger~\cite{Karger93,BenczurK96}. If we allow $c_1, \hdots, c_n \in \mathbb C$ to vary arbitrarily, we recover the notion of spectral sparsification of graphs~\cite{SpielmanT11,SpielmanS11,BatsonSS14}. But our \cref{thm:main} holds for \emph{all} $\psi\in\C^{2^{[n]}}$, and thus is \emph{considerably} more general than these previous two cases. This generality can also be seen very easily from the ``Kikuchi perspective'', since sparsifying the Quantum Cut predicate involves sparsifying \emph{all} the Kikuchi graphs $\cK_{G, k}$ simultaneously for all $0\leq k\leq n$, and because $\cK_{G, 1}$ is isomorphic to $G$, a quantum cut sparsifier has to also spectrally sparsify the graph $G$ at the very least.

Thus given the importance of Kikuchi graphs, we restate \cref{thm:main} in the language of Kikuchi graphs below to highlight the combinatorial nature of the theorem:
\begin{theorem}[Kikuchi sparsifiers]
\label{thm:kikuchimain}
    For any graph $G = ([n], E, w)$ and $\eps > 0$ there exists a (weighted) graph $H = ([n], \tilde{E}, \tilde{w})$ such that $\lvert\supp(\tilde{w})\rvert\leq\tilde{O}(\eps^{-2}n)$, and 
    \[(1 - \eps)\cL_{\cK_{G, k}}\preceq\cL_{\cK_{H, k}}\preceq(1 + \eps)\cL_{\cK_{G, k}}\]
    for all $0\leq k\leq n$. 
\end{theorem}

Given the connections above, henceforth we use the terms ``quantum cut sparsifiers'' and ``Kikuchi sparsifiers'' interchangeably.
\subsection{Challenges in Proving \cref{thm:main}}
It is worth reflecting on \emph{why} \cref{thm:main} is hard to prove, as it may seem that \cite{BasuBP26} has already introduced a robust framework for understanding the sparsifiability of Hamiltonians. Unfortunately, the results of \cite{BasuBP26} come in two forms: 
\begin{enumerate}
    \item The first is a sequence of reductions from Pauli-Hamiltonians to classical CSP sparsification. Such a reduction is not possible in the setting of Quantum Cut: the $2$-qubit Hamiltonian $M_{\mathrm{QC}}$ has rank $1$, and so any attempted reduction would yield a classical CSP where each constraint has $1$ satisfying assignment. But, the theory of classical CSP sparsification \cite{FiltserK17, ButtiZ20, KhannaPS24, BrakensiekG25} shows that it is impossible to get sparsifier sizes below $n^2$ for these types of CSPs! This also highlights an interesting fact in its own right: with \cref{thm:main} in mind, one can draw the conclusion that the entanglement inherent in a QC Hamiltonian actually makes sparsification \emph{easier}. Indeed, without this entanglement, the problem reduces to classical sparsification where lower bounds \emph{rule out} near-linear size sparsifiers.
    \item The second type of result in \cite{BasuBP26} relies on basic invocations of Matrix Chernoff. Roughly speaking, they bound the leverage scores for each local term in the Hamiltonian; in our setting these would be $p_e = \max_{\braket{\psi}} \frac{\langle \psi | \cL_e|\psi \rangle}{\langle \psi | \cL_G|\psi \rangle}$ (that is, the \emph{maximum} fraction of the energy that ever comes from that single term). In \cite{BasuBP26}, they show that for many natural choices of Hamiltonians, many of these leverage scores are bounded below $1/n$; this is crucial for their application. Indeed, this is because they then \emph{randomly subsample} those terms which have small leverage scores, and use Matrix Chernoff to argue that the energy of the Hamiltonian concentrates. But, because the states $\braket{\psi}$ live in $\C^{2^{[n]}}$, applying Matrix Chernoff gives concentration only with probability $1 - 2^n \cdot 2^{-1/R}$, where $R$ is the \emph{largest} leverage score. For any value of $R > \frac{1}{n}$, this statement gives no meaningful bound, as then $2^n \cdot 2^{-1/R} > 1$. 

    This turns out to be a major bottleneck in our setting. Indeed, even for a Quantum Cut Hamiltonian that arises from the complete graph, it is the case that every local term has leverage at least $1/n$. This makes naively applying Matrix Chernoff impossible and is a barrier to even mild subsampling (say, at rate $1/2$).  Similar barriers arise when using, for instance, black-box results for sparsifying sums of PSD matrices (see \cite{BasuKLM26}). Here again, there is an extra logarithmic factor of the dimension which is paid. Because the dimension is $2^n$, this leads to sparsifiers of size $O(n^2 / \eps^2)$, which is trivial (as any graph contains at most $n^2$ edges to begin with).  
\end{enumerate}

As we shall see, \cref{thm:main} relies on techniques far more subtle than \cite{BasuBP26}. We explain these in greater detail below.

\subsection{Our Techniques}

Let us first describe our ideas for unweighted graphs $G$. In this setting, as discussed above, even sparsifying the \emph{complete} graph is non-trivial. 
Our starting point is a recent result~\cite{Kothari25} that, as a special case, yields a near-linear QC-sparsifier for the complete graph. More generally, their work shows that Kikuchi graphs of random $2r$-uniform hypergraphs with appropriate density are spectral approximations to those of complete $2r$-uniform hypergraphs. The proof there is based on two ideas. First, they decompose the action of (the Laplacian of) the level $\ell$ Kikuchi graph (on vertex set ${[n] \choose \ell}$) into a direct sum of ``invariant subspaces'' $E_0 \oplus E_1 \oplus \cdots \oplus E_{\ell}$, each of which is invariant under the action of matrices from the \emph{Johnson scheme}; in particular, for the action of the Kikuchi graph of the complete $2r$-uniform \emph{hypergraph}. They then observe that when restricted to any invariant subspace, the sum of the leverage scores of all the hyperedges is upper-bounded at a scale that just suffices for an application of a matrix concentration inequality to work! Unfortunately this is not enough for their purposes; the subspaces obtained are not invariant under the eigenspaces of Kikuchi graphs of arbitrary (and, in particular, sampled) hypergraphs. To tackle this, their main observation is that the action of Kikuchi graphs on any single $E_i$, while not invariant, is ``band limited''. That is, any $\psi \in V_i$ is sent to $\bigoplus_{j \leq O(r)}V_{i \pm j}$. This allows them to show that applying matrix concentration restricted to direct sums of contiguous blocks of $O(r)$ $V_i$'s suffices to get a sparsifier. 

While this does suffice for the complete graph, here we must deal with \emph{arbitrary} graphs, where the results in~\cite{Kothari25} do not apply.\footnote{Indeed, an appropriate generalization of~\cite{Kothari25} to arbitrary hypergraphs will result in a resolution of the question of finding the correct density vs running time trade-offs for semirandom planted CSPs, a question that has now remained open for a few years.}

\paragraph{A direct proof for complete graphs} Our plan is to begin with an analogous decomposition of the action of the Kikuchi graph into subspaces that, unlike in the case of (higher arity) hypergraphs, are \emph{invariant for all graphs}! This allows us to completely forego the ``band-limited'' analysis mentioned above.

To achieve this, our decomposition is the standard ``harmonic'' decomposition of functions on slices of the $n$-dimensional Boolean hypercube (see, for instance \cite{Grigoriev01, Filmus16}). In this decomposition, there is a subspace for each $0 \leq i \leq n/2$ and the $i^{\text{th}}$ subspace has dimension $\leq n^{i + 2}$. When $G$ is the complete graph, the eigenvalues in each subspace are well-known, and this suffices for showing that the sum of the leverage scores restricted to the $i^{\text{th}}$ invariant subspace is proportional to $n/i$. This is because restricted to the subspace, the operator is in fact constant (i.e., has a single eigenvalue). We can thus apply matrix concentration inequalities and obtain a much more direct argument for sparsifying Kikuchi graphs of \emph{complete graphs}.  

Unfortunately, the above can still dramatically fail for arbitrary graphs $G$; here, the associated operator can have widely varying eigenvalues, even within a single invariant subspace, and thus our argument above for bounding the sum of the leverage scores does not work. This constitutes the primary difficulty in building QC-sparsifiers. To address this problem, it turns out that it is most illustrative to first tackle the case of \emph{expander graphs} before generalizing. 

\paragraph{Dealing with expanders} Since expander graphs are spectral approximations to complete graphs, one may expect the leverage scores to behave similarly to the complete graph, and thus be sufficiently small such that random sampling works. Proving this, it turns out, is quite difficult and we know of no elementary argument so far. 

Our proof instead relies on establishing a new lower bound for the spectrum of QC operators associated with expander graphs in terms of the (appropriately scaled)  QC operator associated with a complete graph. Our key idea is invoking a powerful inequality of Alon and Kozma~\cite{AlonKozma20}. Indeed, in hindsight, one may view this inequality as proving precisely such a spectral lower bound in the more general setting of the so-called \emph{interchange process} (see below). It turns out that this spectral lower bound suffices for bounding the sum of leverage scores on the $i^{\text{th}}$ invariant subspace by $\sim n/i$ --- with this in hand one can show that random sampling (at the appropriately defined rate) succeeds in building QC sparsifiers! 

\paragraph{Our Key Tool: the Alon-Kozma Inequality} We believe that introducing this powerful tool is an important contribution of our work and include a brief overview of the setting that this inequality applies in.

The interchange process is a Markov chain defined on the elements of the symmetric group $\Sn$ where the transitions are associated with a graph $G$ on $[n]$. In the process of interchange for a given graph $G = ([n], E, w)$, when a given element is present $\sigma \in \Sn$, a random edge $e = \{u, v\} \in E$ is sampled proportional to the weights $w$ and the transposition  $(uv)$ is applied to $\sigma$. The Alon-Kozma inequality~ \cite{AlonKozma20} implies that this interchange process operator for an expander graph always dominates the interchange process operator for the complete graph (up to an appropriate scaling). The proof makes crucial use of the so called ``octopus inequality'' developed by Caputo, Liggett, and Richthammer~\cite{CLR10} in their breakthrough resolution of Aldous' conjecture~\cite[Open Problem 14.29]{AldousF02}--see \cref{subsec:related-work} for further details. 

We use this spectral bound for the interchange process to derive a bound for Quantum Cut Hamiltonians by adapting an observation of \cite{CLR10} that relates the interchange process to an \emph{exclusion process} that directly corresponds to a random walk on a Kikuchi graph, see \cref{subsubsec:kikuchiquantumconnection} and \cref{prop:IP-to-EP}. Ultimately, this yields bounds of the form $\cL_{G} \succeq \lambda \cL_{K_n}$, where $\lambda$ is an appropriately defined scaling. From here, it is relatively straightforward to show that in the $i$th invariant subspace, the expander graph inherits a bound on the sum of leverage scores from the complete graph which look like $n/i$.

\paragraph{Generalizing to all graphs } Our idea for generalizing to all graphs is natural and inspired by several works in graph algorithms: \emph{expander decomposition}. These results allow for the decomposition of any graph (with an average degree $\gg \log n$) into pieces that are \emph{mild} expanders. That is, each piece has a spectral gap that grows at least as $1/\polylog(k)$ where $k$ is the number of vertices in the piece while the total number of edges that go between the edges is $O(n\log(n))$. Moreover, each vertex only appears in $O(\log^2(n))$ many expanders. Thus, one can apply the aforementioned sparsification result to each expander, and simply take their union (along with the final, remaining non-expander edges).

A priori, it may seem that the spectral gap of the piece being only $1/\polylog(n)$ is a problem. But it turns out that our argument for expanders above naturally generalizes to such mild expanders with a loss of only a $\polylog(n)$ factor in the size of the resulting sparsifier. This then suffices for sparisfying arbitrary, unweighted graphs.

\paragraph{Handling Weighted Graphs} So far, we have discussed proving \cref{thm:main} for unweighted graphs. For weighted graphs, we can take advantage of a weight-bucketing approach (especially similar to \cite{KhannaPS24, KhannaPS25, BrakensiekG25}) where edges of similar weight are placed in a single bucket (e.g., within a factor of $n^{O(1)}$). For all the edges within a bucket, a simple modification of the unweighted approach results in $\widetilde{O}_{\eps}(n)$-sized sparsifiers. However, to aggregate between buckets, an additional idea is needed. Consider two graphs $G$ and $H$ for which the minimum weight of any edge of $G$ is $n^{\Omega(1)}$-times larger than the maximum weight of any edge of $H$. Using an operator inequality due to Chen~\cite{Chen17} (which also builds on the octopus inequality), we can prove that any edge $e$ of $H$ which connects two vertices in the same component of $G$ is spectrally insignificant (that is, the spectrum of $\cL_G + \cL_H$ is closely approximated by the spectrum of $\cL_{G} + \cL_{H - e}$). Furthermore, two edges of $H$ which connect the same pair of connected components of $G$ can be ``aggregated'' into a single edge without significantly changing the spectrum. As such, once a sparsifier of $G$ is computed, $H$ can be aggregated into a \emph{much smaller graph} whose sparsifier size is proportional to the number of connected components of $G$! The key from this procedure is that each time we sparsify the edges from a single weight bucket, we are promised that in the future, these edges can effectively be ``merged away''. This then reduces the number of connected components in the graph going forward as we sparsify later weight buckets. The key point is that even when aggregating a (potentially) large number of weight buckets, the total sparsifier size \emph{combined across all buckets} will still be near-linear!

\subsection{Related Work}\label{subsec:related-work}

QC sparsification, or equivalently, Kikuchi Sparsification, has natural connections to algorithm design and probability in addition to quantum computation. We include brief pointers to relevant literature here. 

\paragraph{Octopus Inequality.} As previously mentioned, key to the proof of \cref{thm:main} is the use of a spectral inequality due to Alon and Kozma~\cite{AlonKozma20} on the spectral structure of the interchange process. The desire for deeply understanding the spectrum of the interchange process originated due to a conjecture of Aldous~\cite{Aldous,AldousF02}\footnote{Amusingly, Aldous's website~\cite{Aldous} states ``\emph{Instinct says that the same mathematical question will arise in some quite different setting (quantum epistemology?!)}''. We take great pride in reifying the quantum aspect of Aldous's remark. However, we failed to make any connections to epistemology -- we leave this as a stimulating direction for the reader to pursue!} in 1992 that the spectral of the interchange process for a graph $G$ should be identical to the spectral graph of the classical random walk Markov chain on $G$. This conjecture captivated the interest of the probability theory community for decades~\cite{DiaconisS81,DiaconisS93,LevinPW09,FlattoOW85,Cesi10,HandjaniJ96} culminating in a full proof of Aldous' conjecture by Caputo, Liggett, and Richthammer~\cite{CLR10}. The key innovation in \cite{CLR10} was the introduction of a so-called ``octopus inequality'' which can be stated in the language of Quantum Cuts as follows:

\begin{fact}[The octopus inequality~\cite{CLR10} as stated in \cite{AlonKozma20}, translated with \cref{prop:IP-to-EP}]
For any weight function $w : \binom{V}{2} \to \R_{\ge 0}$ and any vertex $u \in V$, we have that
\[
    \sum_{v \in V \setminus \{u\}} w_{uv} \cdot \cL_{uv} \succeq \sum_{\{v, v'\} \in \binom{V \setminus \{u\}}{2}} \frac{w_{uv} w_{uv'}}{\sum_{v'' \in V \setminus \{u\}} w_{uv''}} \cdot \cL_{vv'}.
\]
\end{fact}

Since their introduction by \cite{CLR10}, the octopus inequality has found applications in understanding the subprocesses of the interchange process~\cite{Chen17}, spectral bounds on expanders~\cite{AlonKozma20}, as well as extensions to hypergraphs \cite{AlonKP25}. We hope our new (indirect) application of the octopus inequality to Hamiltonian sparsification can promote new connections between these areas. 

\paragraph{Planted CSPs} Kikuchi sparsification has natural connection to solving (semirandom) planted instances of the $k$-XOR problem. Such an instance is described by a $k$-uniform hypergraph $\cH \subseteq {[n] \choose k}$ along with a bit $b_C$ for each $C \in H$. There is an unknown ``planted assignment'' $x^* \in \{-1,1\}^n$ and each $b_C = \prod_{i \in C} x^*_i$ (i.e., agrees with the value $x^*$ gives to the constraint corresponding to $C$) with probability $1/2 + \eta$ for some fixed constant $\eta>0$ independently. An approach developed in~\cite{GuruswamiHKM23} reduces the analysis of a natural semi-definite programming relaxation for this problem to a sparsification question for Kikuchi graphs of hypergraphs. The level $\ell$ Kikuchi graph of a $2r$-uniform hypergraph $\cH$ is a graph on ${[n] \choose \ell}$ where two vertices $S,T$ are connected iff $S \Delta T \in \cH$. Such Kikuchi graphs (and their variants for the odd $k$) arise in several prior works on related problems and applications. To obtain the optimal conjectured algorithms for the planted $2r$-XOR problem above, we need the level $\ell$ Kikuchi graph of a random sample $\cH'$  of $\cH$ to be a spectral approximation to the Kikuchi graph of $\cH$. When $\ell = r$ (the ``base'' case), such a result was shown to be true in~\cite{GuruswamiHKM23}. For larger $\ell$ it remains open despite efforts, except in the special case when $\cH$ is a random hypergraph; the work of~\cite{Kothari25} implies such a result. 

Our work suggests a natural approach to the problem: prove an appropriate generalization of Alon-Kozma inequality that we crucially rely on in this work for some useful notion of \emph{expanding} hypergraphs. Here useful means that every hypergraph should admit a decomposition into pieces that are ``expanding''. 

\subsection{Outline}

In \cref{sec:prelim}, we set up some preliminaries and establish basic notation. In \cref{sec:harmonicanalysis}, we work out in our notation the harmonic analysis on hypercube slices. In \cref{sec:AlonKozma}, we discuss the spectral inequality of \cite{AlonKozma20} for the interchange process and demonstrate how this connects to Quantum Cut. In \cref{sec:sparsify-expander}, we show how to use the machinery in \cref{sec:harmonicanalysis} and \cref{sec:AlonKozma} to prove a version of \cref{thm:main} for expander graphs. In \cref{sec:expander-decomposition}, we use the technique of expander decomposition to prove \cref{thm:main} for unweighted graphs. In \cref{sec:weighted}, we extend the proof of \cref{thm:main} to weighted graphs.

\subsection*{AI Disclosure}

The development in \cref{sec:sparsify-expander} follows from combining the ideas from~\cite{Kothari25} along with the powerful Alon-Kozma inequality~\cite{AlonKozma20}. The authors learned of this strategy from an answer provided by ChatGPT 5.5 Pro in an interactive prompting session on Kikuchi sparsification for expander graphs. 
The proofs of \cref{sec:expander-decomposition,sec:weighted} were done without the use of AI. The content of this paper is entirely written by the authors. 

\subsection*{Acknowledgments}

A.B. thanks Doron Puder and Irit Dinur for helpful discussions regarding the Alon-Kozma inequality. J.B. thanks Venkatesan Guruswami for valuable conversations and encouragement. A.P. would like to thank Harald Putterman, Fernando Brand\~{a}o, Yeongwoo Hwang, and Anurag Anshu for giving context and motivation for this problem in its early stages. We thank anonymous reviewers for valuable comments which have improved the presentation of this manuscript.

\section{Preliminaries}\label{sec:prelim}

Unless specified otherwise, we let $V$ denote the vertex set and $E$ the edge set of a graph $G$. We further let $n$ and $m$ be shorthand for $|V|$ and $|E|$, respectively.

When describing a vector $v \in \C^S$, where $S$ is a set, we interchangeably use $v(i)$ and $v_i$ to denote the $i$th element of the vector, where $i \in S$.

For any two sets $X, Y$, we denote by $X\oplus Y$ the symmetric difference set, which is $(X\setminus Y)\cup(Y\setminus X)$. We have $X\oplus Y = Y\oplus X$, and for any three sets $X, Y, Z$ we have $(X\oplus Y)\oplus Z = X\oplus(Y\oplus Z):= X\oplus Y\oplus Z$. 

Throughout the paper we shall identify $\binom{[n]}{k}$ with $n$-bit Hamming strings of weight $k$. 

Fix any $0\leq k\leq n$. Let $V_k:= \C^{[n]\choose k}$ be the space of functions $\binom{[n]}{k}\to\C$. We will view $V_k$ as a subspace of $\C^{2^{[n]}}\cong\C^{\zo^n}$, by extending any function $f:\binom{[n]}{k}\to\C$ to a function $f:2^{[n]}\to\C$ by setting $f(S) := 0$ for all $S\in 2^{[n]}\setminus\binom{[n]}{k}$. 

We also define an inner product on $\C^{2^{[n]}}$ by setting 
\begin{equation}
\label{eq:innerproductdef}
    \langle f, g\rangle:= \sum_{S\seq[n]}f(S)^{\dagger}\cdot\overline{g(S)}\mper
\end{equation}
Also define $\|f\|^2:= \langle f, f\rangle$. It is easy to see that $\|f\|^2\geq 0$, with equality iff $f = 0$. 

\begin{fact}[Matrix Chernoff, \cite{Tro15}]
\label{fact:matrix-bernstein-restated}
Let $Y_1, \ldots, Y_m\in\C^{N\times N}$ be independent, random Hermitian matrices. Let $Y = \sum_{i = 1}^m Y_i$ and let $Z = \E[Y]$. Suppose that $Y_i \preceq R \cdot Z$ with probability $1$, where $R$ is an arbitrary scalar, for every $1 \leq i \leq m$. Then, for all $\eps \in [0,1]$,
\[
\Pr\left [(1 - \eps) Z\preceq \sum_{i = 1}^m Y_i \preceq (1 + \eps) Z \right] \geq 1 - 2N \cdot \exp(-\eps^2 / 3R)\mper
\]
\end{fact}

\section{Harmonic Analysis on the Hypercube Slice}
\label{sec:harmonicanalysis}

In this section we develop harmonic analysis on the slices of the boolean hypercube $\binom{[n]}{k}$. We stress that all the material in this section is standard, and we present it here for easy reference. 

We do so by defining the \emph{down} and \emph{up} operators on the slices of the hypercube, and proving various statements about them. The down and up operators are a standard way of presenting this theory, see for instance \cite[Section 2]{Grigoriev01}. Harmonic analysis on the slice of the hypercube is a very well-studied field with connections to multiple areas, see \cite{Dunkl76,Stanley88,Godsil10,Filmus16,DiksteinDFH24} for other perspectives.

\begin{definition}[Down and Up operators]
    For $k\geq 1$, define the \emph{down operator} $\D_k:V_k\to V_{k - 1}$, where for any $f\in V_k$ and $S\in\binom{[n]}{k - 1}$, we define 
    \[(\D_kf)(S):= \sum_{\binom{[n]}{k}\ni T\supset S}f(T)\mper\]
    Similarly, for any $k < n$, define the \emph{up operator} $\U_k:V_k\to V_{k + 1}$, where for any $f\in V_k$ and $S\in\binom{[n]}{k + 1}$, we define 
    \[(\U_kf)(S):= \sum_{\binom{[n]}{k}\ni T\subset S}f(T)\mper\] 
    Finally, for any $0\leq k < \ell\leq n$, define the map $\U^{(\ell)}_k:V_k\to V_\ell$ as $\U^{(\ell)}_k:= \U_{\ell - 1}\circ\cdots\circ\U_{k + 1}\circ\U_k$. We also define $\U^{(k)}_k:V_k\to V_k$ to be the identity map $\Id_{V_k}$.
\end{definition}

These operators allow us to define \emph{harmonic functions}:
\begin{definition}[Harmonic Function]
\label{def:harmonic}
    A function $f\in V_k$ is called \emph{harmonic} if $\D_kf = 0$. More explicitly, if $f$ is harmonic, then for any $T\in\binom{[n]}{k - 1}$, we have 
    \[\sum_{t\in[n]\setminus T}f(T\cup\{t\}) = 0\mper\]
\end{definition}
For $k > 0$, define $H_k:= \ker(\D_k)\subset V_k$ to be the subspace of harmonic functions. For $k = 0$, define $H_0:= V_0$. Finally, for any $0\leq k\leq \lfloor n/2\rfloor$ define the subspace
\begin{equation}
\label{eq:ekdef}
    E_k:= H_k\oplus \U^{(k + 1)}_k(H_k)\oplus\cdots\oplus \U^{(n - k)}_k(H_k)\subset\C^{2^{[n]}}\mper
\end{equation}
Note that the direct sum in \cref{eq:ekdef} is orthogonal, since for any $\ell\neq\ell'\in\{k, \ldots, n - k\}$, functions in $\U^{(\ell)}_k(H_k), \U^{(\ell')}_k(H_k)$ are supported on disjoint sets ($\binom{[n]}{\ell}$ and $\binom{[n]}{\ell'}$ respectively) and are thus orthogonal according to \cref{eq:innerproductdef}.

It is useful to record certain preliminary properties of the down and up operators:
\begin{lemma}
\label{lem:downupproperties}
    For any $k$ such that the relevant symbols are defined, we have: 
    \begin{enumerate}[(1)]
        \item\label{item:adjoint} $\D_k = \U^*_{k - 1}$, i.e. $\D_k$ and $\U_{k - 1}$ are adjoints of each other w.r.t. the inner product in \cref{eq:innerproductdef}.
        \item\label{item:commrel} $\D_{k + 1}\U_k - \U_{k - 1}\D_k = (n - 2k)\cdot\Id_{V_k}$. 
        \item\label{item:injective} For $0\leq k < n/2$, $\U_k$ is injective. Consequently, $\D_k$ is surjective for $1\leq k < n/2 + 1$ by adjointness.
        \item\label{item:iterinj} For any $0\leq k\leq n/2$, $\U^{(\ell)}_k$ is injective on $H_k$ for any $k\leq\ell\leq n - k$. 
    \end{enumerate}
\end{lemma}
\begin{proof}
    To show that $\D_k = \U_{k - 1}^*$, we have to show that for any $f\in V_k, g\in V_{k - 1}$, $\langle \D_kf, g\rangle = \langle f, \U_{k - 1}g\rangle$. To that extent, 
    \[\langle \D_kf, g\rangle = \sum_{T\in\binom{[n]}{k - 1}}(\D_kf)(T)^{\dagger}\cdot\overline{g(T)} = \sum_{T\in\binom{[n]}{k - 1}}\sum_{S\supset T, |S| = k}f(S)^{\dagger}\cdot\overline{g(T)} = \sum_{S\in\binom{[n]}{k}}\sum_{T\subset S, |T| = k - 1}f(S)^{\dagger}\cdot\overline{g(T)}\]
    \[ = \sum_{S\in\binom{[n]}{k}}f(S)^{\dagger}\cdot\sum_{T\subset S, |T| = k - 1}\overline{g(T)} = \sum_{S\in\binom{[n]}{k}}f(S)^{\dagger}\cdot\overline{(\U_{k - 1}g)(S)} = \langle f, \U_{k - 1}g\rangle,\]
    as desired.

    For the second item, fix any $S\in\binom{[n]}{k}$ and note that
    \begin{align}
    \label{eq:updownequation}
(\U_{k - 1}\D_k f)(S) &= \sum_{T \subset S} (\D_k f)(T) = \sum_{T \subset S} \sum_{S' \supset T} f(S') = k \cdot f(S) + \sum_{S' : |S \oplus S'| = 2} f(S'), \\
(\D_{k + 1}\U_k f)(S) &= \sum_{T \supset S} (\U_k f)(T) = \sum_{T \supset S} \sum_{S' \subset T} f(S') = (n - k) \cdot f(S) + \sum_{S' : |S \oplus S'| = 2} f(S'),
\end{align}
as desired.

For the third item, note that $\U_0$ is obviously injective. Thus assume $1\leq k < n/2$, and note that 
\[\|\U_kf\|^2 = \langle\U_kf, \U_kf\rangle \overset{\text{\cref{item:adjoint}}}{=} \langle f, \D_{k + 1}\U_kf\rangle \overset{\text{\cref{item:commrel}}}{=} \langle f, ((n - 2k)\Id + \U_{k - 1}\D_k)f\rangle\] 
\[= (n - 2k)\|f\|^2 + \langle f, \U_{k - 1}\D_kf\rangle \overset{\text{\cref{item:adjoint}}}{=} (n - 2k)\|f\|^2 + \langle \D_kf, \D_kf\rangle = (n - 2k)\|f\|^2 + \|\D_kf\|^2.\]
Consequently, if $\U_kf = 0$, then $\|\U_kf\| = 0\implies\|f\| = 0\implies f = 0$, since $n - 2k > 0$. Consequently, $\U_k$ is injective. 

Finally, consider $\U^{(\ell)}_k$ for any $k\leq\ell\leq n - k$. Note that if $k = n/2$ or $k = \ell$, then $\U^{(\ell)}_k = \Id_{V_k}$ is injective. Similarly, if $k = 0$, then $\U^{(\ell)}_k$ is also injective since $\dim(V_0) = 1$. Thus assume $1\leq k < \min\{\ell, n/2\}$. Fix any $\psi_k\in H_k$, and write $\psi_m:= \U^{(m)}_k\psi_k$ for all $k\leq m\leq n - k$. We claim that for all $k + 1 \leq m\leq n - k$, $\D_m\psi_m = c_m\psi_{m - 1}$, for some $c_m > 0$. Note that this proves that $\U^{(\ell)}_k$ is injective since 
\[\|\U^{(\ell)}_k\psi_k\|^2 = \langle\psi_\ell, \psi_\ell\rangle = \langle\U_{\ell - 1}\psi_{\ell - 1}, \psi_{\ell}\rangle \overset{\text{\cref{item:adjoint}}}{=} \langle\psi_{\ell - 1}, \D_\ell\psi_{\ell}\rangle = c_\ell\langle\psi_{\ell - 1}, \psi_{\ell - 1}\rangle = c_\ell\|\U^{(\ell - 1)}_k\psi_k\|^2\]
\[\implies\|\U^{(\ell)}_k\psi_k\|^2 = \left(\prod_{m = k + 1}^\ell c_m\right)\cdot\|\psi_k\|^2,\]
and thus if $\U^{(\ell)}_k\psi_k = 0$ then $\psi_k = 0$ since $\prod_{m = k + 1}^\ell c_m > 0$. Thus, it only remains to prove that $\D_m\psi_m = c_m\psi_{m - 1}$, which we prove inductively. For convenience set $c_k:= 0$ (since $\D_k\psi_k = 0$) and $\psi_{k'}:= 0$ for all integers $k' < k$. Then for any $m\geq k + 1$,  
\[\D_m\psi_m = \D_m\U_{m - 1}\psi_{m - 1} \overset{\text{\cref{item:commrel}}}{=} \left((n - 2(m - 1))\Id + \U_{m - 2}\D_{m - 1}\right)\psi_{m - 1}\] 
\[= (n - 2m + 2)\psi_{m - 1} + \U_{m - 2}\cdot(c_{m - 1}\psi_{m - 2}) = (n - 2m + 2)\psi_{m - 1} + c_{m - 1}\psi_{m - 1},\]
and thus we get the recurrence relation $c_m:= n - 2m + 2 + c_{m - 1}$. Unrolling the recurrence yields $c_m = (m - k)(n - m - k + 1)$, which is strictly positive for $k + 1\leq m\leq n - k$, as desired.
\end{proof}

\begin{proposition}\label{prop:boundDimension}
 For any $0\leq k\leq \lfloor n/2\rfloor$, we have $\dim(E_k) = (n - 2k + 1)\left(\binom{n}{k} - \binom{n}{k - 1}\right)$, where we set $\binom{n}{-1}:= 0$. 
\end{proposition}
\begin{proof}
    We have $\dim(E_k) = \sum_{\ell = k}^{n - k}\dim(\U^{(\ell)}_k(H_k)) \overset{(\ast)}{=} \sum_{\ell = k}^{n - k}\dim(H_k) = (n - 2k + 1)\cdot\dim(H_k)$, where $(\ast)$ follows from the fact that $\U^{(\ell)}_k$ is injective on $H_k$ for any $k\leq\ell\leq n - k$ (\cref{item:iterinj} of \cref{lem:downupproperties}). Thus it suffices to show that $\dim(H_k) = \binom{n}{k} - \binom{n}{k - 1}$. But this follows from the fact that $\D_k$ is surjective (\cref{item:injective} of \cref{lem:downupproperties}). 
\end{proof}

\begin{lemma}
\label{lem:orthodecomp}
    We have the orthogonal (w.r.t. \cref{eq:innerproductdef}) decomposition 
    \[\C^{2^{[n]}} = \bigoplus_{k = 0}^{\lfloor n/2\rfloor}E_k\mper\]
\end{lemma}
\begin{proof}
    Note that it suffices to show that $E_k$ and $E_{k'}$ are orthogonal spaces for $k\neq k'$, since by \cref{prop:boundDimension} we have $\sum_{k = 0}^{\lfloor n/2\rfloor}\dim(E_k) = 2^n = \dim(\C^{2^{[n]}})$. Furthermore, since the direct sum in \cref{eq:ekdef} is orthogonal, it suffices to show that for any $\ell\geq k' > k$, $\U^{(\ell)}_k(H_k)$ and $\U^{(\ell)}_{k'}(H_{k'})$ are orthogonal spaces. Thus fix $\psi_k\in H_k, \psi_{k'}\in H_{k'}$, and consider 
    \[\langle\U^{(\ell)}_k\psi_k, \U^{(\ell)}_{k'}\psi_{k'}\rangle = \langle\U_{\ell - 1}\circ\cdots\circ\U_k\psi_k, \U_{\ell - 1}\circ\cdots\circ\U_{k'}\psi_{k'}\rangle\]
    \begin{equation}
    \label{eq:downuportho}
        \overset{\text{\cref{item:adjoint} of \cref{lem:downupproperties}}}{=}\langle\psi_k, \D_{k + 1}\circ\cdots\circ\D_\ell\circ\U_{\ell - 1}\circ\cdots\circ\U_{k'}\psi_{k'}\rangle\mper
    \end{equation}
    Now, repeatedly apply \cref{item:commrel} of \cref{lem:downupproperties} to the string $\D_{k + 1}\circ\cdots\circ\D_\ell\circ\U_{\ell - 1}\circ\cdots\circ\U_{k'}$ in the following way: Take the rightmost pair of $\D\circ\U$ and use \cref{item:commrel} of \cref{lem:downupproperties} to obtain $\U\circ\D$ and a multiple of $\Id$. This way we produce $2^{\ell - k'}$ strings all of which end in $\D_{k'}$: This is because there are more $\D$s than $\U$s (since $k < k'$) and thus we can consume all the $\U$s \footnote{``pacman'' style} and only be left with $\D$s. But since $\psi_{k'}\in H_{k'}$, $\D_{k'}\psi_{k'} = 0$, and thus 
    \[\langle\psi_k, \D_{k + 1}\circ\cdots\circ\D_\ell\circ\U_{\ell - 1}\circ\cdots\circ\U_{k'}\psi_{k'}\rangle = 0,\]
    as desired.
\end{proof}

\section{Kikuchi Graphs and the Symmetric Group}
\label{sec:Kikuchisym}
In this section we develop yet another point of view on the operator $\cL_G$ (defined in \cref{eq:qcpredicate}) which will be useful for us later on.

Towards that end, let $\Sn$ be the group of all permutations on $n$ alphabets. $\Sn$ acts naturally on $2^{[n]}$, where $\pi\in\Sn$ sends $S\in 2^{[n]}$ to $\pi S:= \{\pi(s): s\in S\}$. Also define the linear operator $A_\pi:\C^{2^{[n]}}\to\C^{2^{[n]}}$, where for any $\psi\in\C^{2^{[n]}}$ and $S\in 2^{[n]}$ we define
\[(A_\pi\psi)(S):= \psi(\pi S)\mper\]
Note that $A_\pi(V_k)\seq V_k$ for all $0\leq k\leq n$, since $|\pi S| = |S|$ for any $S\in 2^{[n]}$. Also note that $A_\pi$ is unitary for any $\pi\in\Sn$. Now, for any $e = \{i, j\}\in\binom{[n]}{2}$, denote by $(e):= (i j)$ the transposition swapping $i$ and $j$. Note that 
\[(e)S = \begin{cases}
    S & \text{if }|S\cap e|\in\{0, 2\}, \\
    S\oplus e & \text{if }|S\cap e| = 1
\end{cases}.\]

We claim that $A_e$ is Hermitian, where we simplify notation to write $A_{(e)}$ as $A_e$.
\begin{proposition}
\label{prop:Aeeig}
    For any $e = \{i, j\}\in\binom{[n]}{2}$, $A_e$ is Hermitian and unitary. Consequently, all its eigenvalues are $\pm 1$. 
\end{proposition}
\begin{proof}
    Fix any $\psi, \phi\in\C^{2^{[n]}}$. Then 
    \[\langle\psi, A_e\phi\rangle = \sum_{S\in 2^{[n]}}\psi(S)^{\dagger}\cdot\overline{\phi((e)S)} = \sum_{|S\cap e|\in\{0, 2\}}\psi(S)^{\dagger}\cdot\overline{\phi(S)} + \sum_{|S\cap e| = 1}\psi(S)^{\dagger}\cdot\overline{\phi(S\oplus e)},  \]
    \[\langle A_e\psi, \phi\rangle = \sum_{S\in 2^{[n]}}\psi((e)S)^{\dagger}\cdot\overline{\phi(S)} = \sum_{|S\cap e|\in\{0, 2\}}\psi(S)^{\dagger}\cdot\overline{\phi(S)} + \sum_{|S\cap e| = 1}\psi(S\oplus e)^{\dagger}\cdot\overline{\phi(S)}.\]
    But note that $S\mapsto S\oplus e$ is a bijection on the set $\{S\in 2^{[n]}: |S\oplus e| = 1\}$ since $(S\oplus e)\oplus e = S$, and thus $\langle\psi, A_e\phi\rangle = \langle A_e\psi, \phi\rangle$, thus showing that $A_e$ is Hermitian, as desired.
\end{proof}

We now claim that $\Id - A_e$ is exactly the operator $\cL_e$ defined in \cref{eq:qcpredicate} and \cref{subsubsec:kikuchiquantumconnection}. Indeed, for any $\psi\in V_k$, we have 
\[\langle\psi, \cL_e\psi\rangle = \sum_{\{S, S\oplus e\}\in E(\cK_e)}|\psi(S) - \psi(S\oplus e)|^2\mper\]
On the other hand, 
\[\langle\psi, (\Id - A_e)\psi\rangle = \sum_{S\in\binom{[n]}{k}}\psi(S)^{\dagger}\cdot\overline{(\psi(S) - (A_e\psi)(S))} = \sum_{S\in\binom{[n]}{k}:|S\cap e| = 1}\psi(S)^{\dagger}\cdot\overline{(\psi(S) - \psi(S\oplus e))}\]
\[ = \sum_{\{S, S\oplus e\}\in E(\cK_e)}\left(\psi(S)^{\dagger}\cdot\overline{(\psi(S) - \psi(S\oplus e))} + \psi(S\oplus e)^{\dagger}\cdot\overline{(\psi(S\oplus e) - \psi(S))}\right)\] 
\[= \sum_{\{S, S\oplus e\}\in E(\cK_e)}\left|\psi(S) - \psi(S\oplus e)\right|^2 = \langle\psi, \cL_e\psi\rangle,\]
thus showing that $\cL_e = \Id - A_e$.\footnote{Note that $A_e$ is not exactly the adjacency matrix of $\cK_e$: Instead, $A_e$ is the adjacency matrix of $\cK_e'$, where $\cK'_e$ is the graph which has a self-loop for all vertices with zero degree in $\cK_e$. Nevertheless, $\cL_e$ \emph{is} the Laplacian for $\cK_e$ since the $1$ in the diagonal entry (corresponding to a degree $0$ vertex in $\cK_e$) of $\Id$ cancels the $1$ in the diagonal of $A_e$ representing the self-loop}

We'll now use this interpretation of $\cL_e$ to prove that the spaces $E_k$ are left invariant by $\cL_e$ for any $e\in\binom{[n]}{2}$, and then we shall subsequently also compute the eigenvalues of $\cL_{K_n}$, where $K_n$ is the complete unweighted graph on $n$ vertices.

Now, for any $e = \{i, j\}\in\binom{[n]}{2}$ define the linear operator $\cL_e:\C^{2^{[n]}}\to\C^{2^{[n]}}$ as $\cL_e:= \Id - A_e$. Note that by \cref{prop:Aeeig} $\cL_e$ is a Hermitian PSD operator all of whose eigenvalues are $0$ or $2$. 

For any weighted graph $G = ([n], E, w)$, define $\cL_G:= \sum_{e\in E(G)}w_e\cL_e$. 

We now show that the spaces $E_k$ defined in \cref{sec:harmonicanalysis} are left invariant under $\cL_e$. 

\begin{lemma}[Invariant Spaces]\label{lem:invariantSpaces}
    For any $0\leq k\leq \lfloor n/2\rfloor$, and any $e\in\binom{[n]}{2}$, we have $\cL_e(E_k)\seq E_k$. 
\end{lemma}
\begin{proof}
    It suffices to show that $A_e(E_k)\seq E_k$. We first claim that $A_e$ commutes with $\U_k$ for any $0\leq k < n$: Indeed for any $\psi\in V_k$ and $S\in\binom{[n]}{k}$, we have
    \[(A_e\U_k\psi)(S) = (\U_k\psi)((e)S) = \sum_{T\subset(e)S}\psi(T) \overset{(\ast)}{=} \sum_{T\subset S}\psi((e)T) = \sum_{T\subset S}(A_e\psi)(T) = (\U_kA_e\psi)(S)\mcom\]
    where $(\ast)$ uses the fact that $(e)^2 = \id_{\Sn}$. Thus it suffices to show that $A_e(H_k)\seq H_k$ for all $0\leq k\leq \lfloor n/2\rfloor$. Since $H_0 = V_0$ and since $A_e(V_k)\seq V_k$ for any $k$, we're done for $k = 0$. Thus assume $k = 1$, and recall that $\psi_k\in H_k$ iff $\sum_{t\in[n]\setminus T}\psi_k(T\cup\{t\}) = 0$ for all $T\in\binom{[n]}{k - 1}$. Now, fix any $T\in\binom{[n]}{k - 1}$ and note that 
    \[\sum_{t\in[n]\setminus T}(A_e\psi_k)(T\cup\{t\}) = \sum_{t\in[n]\setminus T}\psi_k((e)(T\cup\{t\})) = \sum_{s\in[n]\setminus (e)T}\psi_k((e)T\cup\{s\}) = 0, \]
    and thus $A_eH_k\seq H_k$, as desired.
\end{proof}

Finally, let $K_n = \left([n], \binom{[n]}{2}\right)$ denote the complete (unweighted) graph on $n$ vertices.

\begin{lemma}\label{lem:completeGraphEigenvalues}
    For any $0\leq k\leq \lfloor n/2\rfloor$, we have $\cL_{K_n}\vert_{E_k} = k(n + 1 - k)\cdot\Id_{E_k}$. Consequently, the decomposition in \cref{lem:orthodecomp} is the eigendecomposition of $\cL_{K_n}$. 
\end{lemma}
\begin{proof}
    Note that $E_0$ is the subspace of all functions $\psi\in\C^{2^{[n]}}$ such that $\psi\vert_{V_k}$ is a constant function for any $0\leq k\leq n$.\footnote{$\psi\vert_{V_k}$ and $\psi\vert_{V_{k'}}$ need not assume the \emph{same} constant value for $k\neq k'$ though} For such functions the assertion can be seen to be true by direct computation. Thus assume $k\geq 1$. Fix any $\psi_k\in V_k$, and note that (as in \cref{eq:updownequation}) we have 
    \begin{equation}
        \label{eq:udpsis}
        (\U_{k - 1}\D_k\psi)(S) = k\psi(S) + \sum_{S':|S\oplus S'| = 2}\psi(S') = k\psi(S) + \sum_{x\in S, y\in [n]\setminus S}\psi(S\oplus\{x, y\})\mper
    \end{equation}
    On the other hand, we also have 
    \[(\cL_{K_n}\psi)(S) = \sum_{e\in\binom{[n]}{2}}(\cL_e\psi)(S) = \sum_{e\in\binom{[n]}{2}}(\psi(S) - \psi(S\oplus e)) = \sum_{e:|S\cap e| = 1}(\psi(S) - \psi(S\oplus e))\]
    \[ = k(n - k)\psi(S) - \sum_{x\in S, y\in[n]\setminus S}\psi(S\oplus\{x, y\}) \overset{\text{\cref{eq:udpsis}}}{=} k(n - k + 1)\psi(S) - (\U_{k - 1}\D_k\psi)(S),\]
    and thus $\cL_{K_n}\vert_{V_k} = k(n - k + 1)\Id_{V_k} - \U_{k - 1}\D_k$ for all $1\leq k\leq n$. \footnote{recall that since $A_eV_k\seq V_k$ for all $e\in\binom{[n]}{2}$, we have $\cL_{K_n}V_k\seq V_k$} 

    Now, to show that $\cL_{K_n}\vert_{E_k} = k(n + 1 - k)\cdot\Id_{E_k}$ it suffices to show that $\cL_{K_n}\vert_{\U_k^{(\ell)}(H_k)} = k(n + 1 - k)\cdot\Id_{\U_k^{(\ell)}(H_k)}$ for all $k\leq\ell\leq n - k$. Thus fix any $k \leq \ell\leq n - k$. We know from the proof of \cref{item:iterinj} in \cref{lem:downupproperties} that for any $\psi_k\in H_k$ we have $\D_\ell\U^{(\ell)}_k\psi_k = c_\ell\U^{(\ell - 1)}_k\psi_k$, where $c_\ell:= (\ell - k)(n - \ell - k + 1)$. Thus 
    \[\U_{\ell - 1}\D_\ell\U^{(\ell)}_k\psi_k = c_\ell\U_{\ell - 1}\U^{(\ell - 1)}_k\psi_k = c_\ell\U^{(\ell)}_k\psi_k\implies\cL_{K_n}\U^{(\ell)}_k\psi_k = \left(\ell(n - \ell + 1) - c_\ell\right)\U^{(\ell)}_k\psi_k\]
    \[ = k(n - k + 1)\U^{(\ell)}_k\psi_k\mcom\]
    i.e. $\cL_{K_n}\cdot(\U^{(\ell)}_k\psi_k) = k(n - k + 1)\cdot(\U^{(\ell)}_k\psi_k)$ for all $\psi_k\in H_k$, as desired.
\end{proof}

\section{Applying Spectral Bounds for the Interchange Process}\label{sec:AlonKozma}

To understand the spectral properties of $\mathcal L_G$ for an arbitrary graph $G = ([n], E, w)$, we find it useful to look at more general operators related to the \emph{interchange process} on the symmetric group $\Sn$ ~\cite{AldousF02,DiaconisS81,CLR10}. 

\subsection{The Interchange Process}

In this section, we adapt the notation of Alon and Kozma~\cite{AlonKozma20}. We let $\mathbb C[\Sn]$ denote the ring of formal sums of the form $\sum_{\sigma \in \Sn} c_{\sigma} \sigma$ where each $c_{\sigma} \in \mathbb C$ is arbitrary. Multiplication in $\C[\Sn]$ is defined as follows
\[
    \left(\sum_{\sigma \in \Sn} c_{\sigma} \sigma\right)\left(\sum_{\sigma \in \Sn} c'_{\sigma} \sigma\right) = \sum_{\sigma, \sigma' \in \Sn} c_{\sigma} c'_{\sigma'} (\sigma \sigma'),
\]
where $\sigma \sigma'$ is multiplication in $\Sn$. We can think of multiplication as an operator in the function space $L^2(\Sn)$ (i.e., the space of functions $f : \Sn \to \mathbb C$ with the $L^2$ norm), where $L = \sum_{\sigma \in \Sn} c_{\sigma} \sigma \in \mathbb C[\Sn]$ acts on $f \in L^2(\Sn)$ via
\[
    (Lf)(\sigma') = \sum_{\sigma \in \Sn} c_{\sigma} f(\sigma \sigma').
\]
As such, we can say that an operator $L \in \mathbb C[\Sn]$ is PSD (i.e., $L \succeq 0$) if $\langle f, Lf\rangle \geq 0$ for all $f \in L^2(\Sn)$.

For any edge $e = \{i,j\} \in \binom{[n]}{2}$, we define $\nabla_{e} = \id - (ij) \in \mathbb C[\Sn]$, where $\id$ is the identity permutation and $(ij)$ is the transposition between $i$ and $j$. Using these $\nabla_{e}$ operators, we can build a corresponding operator $\nabla_G$ for any weighted graph $G = ([n], E, w)$ as follows
\[
    \nabla_G := \sum_{e \in E} w_e \nabla_e.
\]

Crucially, any spectral inequalities on $\nabla_G$ translate to spectral inequalities on $\cL_{G}$.

\begin{proposition}\label{prop:IP-to-EP}
Let $G = ([n], E, w)$ and $H = ([n], E', w')$ be weighted graphs and $\lambda \geq 0$ a scalar. If $\nabla_{G} \succeq \lambda \nabla_{H}$ then $\cL_{G} \succeq \lambda \cL_{H}$.
\end{proposition}

\begin{proof}
Assume that $\nabla_G \succeq \lambda \nabla_H$, we seek to prove that $\cL_{G} \succeq \lambda \cL_{H}$. To do this, it suffices to show for any $f \in \C^{2^{[n]}}$ that
\[
    \langle f, (\cL_G - \lambda \cL_H)f \rangle \geq 0.
\]
As a first observation, write $f = f_0 + f_1 + \cdots + f_n$ where $f_k \in V_k = \C^{\binom{[n]}{k}}$ for all $k \in \{0, 1, \hdots, n\}$. Then, observe that
\[
    \langle f, (\cL_G - \lambda \cL_H)f \rangle = \sum_{k=0}^n \langle f_k, (\cL_G - \lambda \cL_H)f_k \rangle,
\]
so it suffices to prove that $\langle f_k, (\cL_G - \lambda \cL_H)f_k \rangle \geq 0$ for all $k \in [n]$. This later fact follows from an observation of Caputo, Liggett, and Richthammer~\cite{CLR10} that the \emph{exclusion process} is a subprocess of the interchange process. That is, $\cL_{G} |_{V_k}$ is a projection of $\nabla_G$. More explicitly, given $f_k \in V_k$ define $g_k \in L^2(\Sn)$ by
\[
    g_k(\sigma) = c_kf_k(\{\sigma(1), \hdots, \sigma(k)\}),
\]
where $c_k = \frac{1}{k!(n-k)!}$. Then, for any $e = \{i,j\} \in \binom{[n]}{2}$, we have that
\begin{align*}
    \langle g_k, \nabla_e g_k\rangle &= \sum_{\pi \in \Sn} g_k(\pi)^{\dagger} (\nabla_eg_k)(\pi)\\
    &= \sum_{\pi \in \Sn} g_k(\pi)^{\dagger} \left[g_k(\pi) - g_k((ij)\pi)\right]\\
    &= \sum_{\pi \in \Sn} c^2_kf_k(\{\pi(1), \hdots, \pi(k)\}) \left[f_k(\{\pi(1), \hdots, \pi(k)\}) - f_k((ij)\{\pi(1), \hdots, \pi(k)\})\right],
\end{align*}
where $(ij)\{\pi(1), \hdots, \pi(k)\}$ is the action of the transposition $(ij)$ on the latter set. Observe for any set $S \in \binom{[n]}{k}$ there are exactly $k!(n-k)! = 1/c_k$ choices of $\pi \in \Sn$ with $S = \{\pi(1), \hdots, \pi(k)\}.$ Therefore,  the above sum can be simplified to
\begin{align*}
\langle g_k, \nabla_e g_k\rangle &= \sum_{S \in \binom{[n]}{k}} f_k(S)^{\dagger}(f_k(S) - f_k((ij)S))\\
&= \langle f_k, \cL_e f_k\rangle.
\end{align*}
To finish, observe that
\begin{align*}
\langle f_k, (\cL_G - \lambda \cL_H)f_k \rangle &= \sum_{e \in \binom{[n]}{2}} \langle f_k, (w_e - \lambda w'_e)\cL_e f_k\rangle\\
&= \sum_{e \in \binom{[n]}{2}} \langle g_k, (w_e - \lambda w'_e)\nabla_e g_k\rangle\\
&= \langle g_k, (\nabla_G - \lambda \nabla_H)g_k \rangle \geq 0,
\end{align*}
as desired.
\end{proof}

\subsection{Relating Expanders to the Complete Graph}

A recent work by Alon and Kozma~\cite{AlonKozma20}, building off of the work of Caputo, Liggett, and Richthammer~\cite{CLR10}, showed that $\nabla_{G} \succeq \lambda \nabla_H$ often holds when $G$ is an expander and $H$ is the complete instance $K_n$. To discuss this relationship, we need to define a discrete Markov chain considered by \cite{AlonKozma20}. For a graph $G = ([n], E, w)$, for any $i \in [n]$ define $w_i := \sum_{e \in E, i \in e} w_e$. Then, let $M_G \in \R^{n \times n}$ be the transition matrix with probabilities
\[
    (M_G)_{i,j} := \begin{cases}
    \frac{w_{ij}}{2w_i} & \{i,j\} \in E\\
    \frac{1}{2} & i = j\\
    0 & \text{otherwise}.
        \end{cases}
\]
Let $\pi_G \in \R^n$ be the unique stationary measure corresponding to $M_G$. That is, $\pi$ is the unique probability distribution for which $\pi_GM_G = \pi_G$. In fact, one can show in this case that $\pi_G(i) = \frac{w_i}{\sum_{i=1}^n w_i}$ for all $i \in [n]$. Alon and Kozma define the discrete mixing time $\dmix G$ to be
\[
    \dmix G := \min\left\{t \in \mathbb N \mid (M_G^t)_{i,j} > \frac{3}{4}\pi(j)\right\}.
\]
We also define an ``error'' parameter $\delta_G$ as follows (see equation 4 of \cite{AlonKozma20})
\[
    \frac{1}{\delta_G} := \prod_{k=0}^{\lfloor \log_2 \dmix G\rfloor} \max_{i \in [n]} \left[1 + (M_G^{2^k})_{i,i}\right].
\]

\begin{theorem}[Theorem 1~\cite{AlonKozma20}]\label{thm:AK}
There exists a universal constant $\cAK > 0$ such that for any $n \in \mathbb N$ and any graph $G = ([n], E, w)$, we have that
\[
    \nabla_G \succeq \frac{\cAK\delta_G}{\dmix G} \cdot \frac{\min_{i \in [n]} w_i^2}{\sum_{i \in [n]} w_i} \nabla_{K_n}.
\]
Furthermore, $\delta_G \geq \frac{1}{2\dmix G}$.
\end{theorem}

\subsection{From Conductance to Discrete Mixing}

Given a weighted graph $G = (V, E, w)$, we define $\Vol_G(S) = \sum_{i \in S} w_i$, where we recall that $w_i = \sum_{e \in E, i \in e} w_e$. Further define
\begin{align}
    \alpha_G = \frac{\min_{i \in n} w_i}{\frac{1}{n}\sum_{i=1}^n w_i}\label{eq:alpha}
\end{align}
to be the ratio between the minimum degree and average degree of $G$. We say that $G$ is a $\phi$-conductance expander if for every cut $S \subseteq V$ for which $\Vol_G(S) \leq \Vol_G(V) / 2$, we have that
\begin{align}
    \frac{w(S, \bar{S})}{\mathrm{Vol}_G(S)} \geq \phi,\label{eq:cond}
\end{align}
where $w(A, B) = \sum_{i \in A} \sum_{j \in B} w_{ij}$. Our goal is to show the following.

\begin{lemma}\label{lem:cond-dmix}
Assume that $G = (V, E, w)$ is a $\phi$-conductance expander. Then, $\dmix G \leq \frac{8+8\log (n / \alpha_G)}{\phi^2}$.
\end{lemma}

The proof of \cref{lem:cond-dmix} proceeds by using a fact of Morris and Peres~\cite{MP05} which connects the discrete mixing time of a Markov chain with its conductance. First, define a quantity for all $S \subseteq [n]$ by
\[
    M_G(S, \bar{S}) := \sum_{i \in S} \sum_{j \in \bar{S}} \pi_G(i)(M_G)_{i,j}
\]
to measure the density of edges leaving $S$. Let $\pi^*_G := \min_{i \in [n]} \pi_G(i)$. For $r \in [\pi^*_G, 1/2]$ define
\[
    \phi_G(r) := \inf\left\{\frac{M_G(S, \bar{S})}{\pi_G(S)} \mid S \subseteq [n], \pi_G(S) \leq r\right\}.
\]
Further, let $\phi_G(r) = \phi_G(1/2)$ for all $r \geq 1/2$. Then, using the fact that $(M_G)_{i,i} \geq 1/2$ for all $i \in [n]$, we have the following result due to Morris and Peres~\cite{MP05}.

\begin{theorem}[Theorem 1~\cite{MP05}]\label{thm:MP}
For any $i, j \in [n]$ and any $\eps > 0$ we have that
\[
    t \geq 1 + \int_{4 \min(\pi(i), \pi(j))}^{4/\eps} \frac{4dr}{r\phi_G^2(r)} \implies \left|\frac{(M_G^t)_{i,j} - \pi_j}{\pi_j}\right| \leq \eps.
\]
\end{theorem}

We can now prove \cref{lem:cond-dmix}.

\begin{proof}[Proof of \cref{lem:cond-dmix}]
First, observe that
\begin{align*}
\frac{M_G(S, \bar{S})}{\pi_G(S)} &= \frac{1}{\pi_G(S)}\sum_{i \in S} \sum_{j \in \bar{S}} \pi_G(i)(M_G)_{i,j}\\
    &= \frac{\sum_{i=1}^n w_i}{\sum_{i \in S} w_i}\sum_{i \in S} \sum_{j \in \bar{S}} \frac{w_i}{\sum_{i=1}^n w_i} \cdot \frac{w_{i,j}}{2w_i}\\
    &= \frac{1}{2\sum_{i \in S} w_i} w(S, \bar{S})\\
    & = \frac{1}{2} \cdot \frac{w(S, \bar{S})}{\Vol_G(S)}.
\end{align*}
Furthermore, $\frac{\Vol_G(S)}{\Vol_G(V)} = \pi_G(S)$. Therefore, by \cref{eq:cond}, we have that $\phi_G(r) \geq \frac{\phi}{2}$ for all $r \in [\pi_G^*, 1/2]$. Observe by simple algebra for any $t \in \mathbb N$ the implication that
\[
\left|\frac{(M_G^t)_{i,j} - \pi_j}{\pi_j}\right| \leq \frac{1}{5} \implies (M_G^t)_{i,j} > \frac{3}{4}\pi_{i,j}.
\]
Thus, if we apply \cref{thm:MP} with $\eps = 1/5$ we get that
\begin{align*}
    \dmix G &\leq \max_{i,j}\left[2 + \int_{4 \min(\pi(i), \pi(j))}^{4/\eps} \frac{4dr}{r\phi_G^2(r)}\right]\\
    &\leq 2 + \int_{4\pi^*_G}^{20} \frac{16dr}{r \phi^2}\\
    &\leq 2 + \frac{2}{\phi^2} \left[\log r\right]_{4\pi^*_G}^{20}\\
    &\leq 2 + \frac{2(\log 20 - 4\log \pi^*_G)}{\phi^2}\\
    &\leq \frac{8 + 8 \log \frac{1}{\pi^*_G}}{\phi^2},
\end{align*}
where the last line uses the fact that $\phi \leq 1$. Now, note that $\pi^*_G = \frac{\alpha_G}{n}$. Thus, $\dmix G \leq \frac{8 + 8 \log (n/\alpha_G)}{\phi^2}$, as desired.
\end{proof}

\subsection{A Spectral Bound for Expanders}

By combining the results from the previous sections, we can conclude with our key spectral bound.

\begin{theorem}\label{thm:expander-Kikuchi}
Assume that $G = (V, E, w)$ with $|V| = n$ is a $\phi$-conductance expander with $\alpha_G$ the ratio between minimum and average degree. Then, 
\[
    \cL_G \succeq \frac{\cAK \phi^4 \alpha_G^2}{128(1 + \log (n/\alpha_G))^2} \cdot \frac{\Vol_G(V)}{n^2}\cdot \cL_{K_n}.
\]
\end{theorem}

\begin{proof}
By \cref{lem:cond-dmix}, we have that $\dmix G \leq \frac{8+8\log (n / \alpha_G)}{\phi^2}$. If we apply \cref{thm:AK} with the bound $\delta_G \geq \frac{1}{2\dmix G}$, we get that
\begin{align*}
\nabla_G &\succeq \frac{\cAK}{2(\dmix G)^2} \cdot \frac{\min_{i \in [n]} w_i^2}{\sum_{i \in [n]} w_i} \nabla_{K_n}.\\
&\succeq \frac{\cAK \phi^4}{128(1 + \log (n/\alpha_G))^2} \cdot \frac{\frac{\alpha_G^2}{n^2} \left(\sum_{i=1}^n w_i\right)^2}{\sum_{i=1}^n w_i} \nabla K_n\\
&= \frac{\cAK \phi^4 \alpha_G^2}{128(1 + \log (n/\alpha_G))^2} \cdot \frac{\sum_{i = 1}^n w_i}{n^2}\cdot \nabla_{K_n}.
\end{align*}
To finish, we use \cref{prop:IP-to-EP} to replace $\nabla_G$ and $\nabla_{K_n}$ with $\cL_G$ and $\cL_{K_n}$, respectively.
\end{proof}

\section{Sparsifying Expanders}\label{sec:sparsify-expander}

In this section, we consider sparsifying the Kikuchi-graph of unweighted expanders $G = (V, E)$. 
Inspired by the above, we will assume that our graph $G= (V, E)$ is a good \emph{conductance expander} and that the minimum degree of any vertex is within a logarithmic factor of the \emph{average degree}.

\begin{lemma}\label{lem:sparsifyExpander}
Let $G = (V, E)$ be a graph on $n' \leq n$ vertices and $m$ edges such that $G$ is a $\phi$ conductance expander and has minimum-degree balance parameter $\alpha_G$. Let $\eps \in (0, 1)$, let $m' = \frac{100 n' \log(n)}{\lambda \eps^2}$ for $\lambda =  \frac{\cAK \phi^4 \alpha_G^2}{128(1 + \log (n'/\alpha_G))^2}$, and let $p = \min  \left (1,  \frac{m'}{m} \right )$. Consider the random process of sampling each edge $e \in E$ with probability $p$, and giving weight $1/p$ to each sampled edge. Denote this resulting graph by $\widetilde{G}$. Then, with probability $\geq 1 - \frac{1}{n^8}$ over the random sampling, it is the case that 
\[
(1 - \eps) \cdot \mathcal{L}_G \preceq \mathcal{L}_{\widetilde{G}} \preceq (1 + \eps) \cdot \mathcal{L}_G,
\]
and $\widetilde{G}$ has $\leq 2m'$ many edges.
\end{lemma}

\begin{proof}
Note that if $p = 1$ or $n = 1$, the claim holds trivially. Thus, we assume $p < 1$ and $n \geq 2$ going forward.

For each edge $e \in E$, we let $X_e$ denote a random variable such that:
\[
X_e = \begin{cases}
    \frac{1}{p} \cdot \mathcal{L}_e  \text{ with probability } p \\
    0 \quad \quad \text{   otherwise}
\end{cases}.
\]
Note that by \cref{prop:Aeeig} all of the eigenvalues of $\mathcal{L}_e = \Id - A_e$ lie in $[0,2]$.

Now, recall that our goal is to show \[
(1 - \eps) \cdot \mathcal{L}_G \preceq \mathcal{L}_{\widetilde{G}} \preceq (1 + \eps) \cdot \mathcal{L}_G.
\]
By \cref{lem:orthodecomp} and \cref{lem:invariantSpaces}, it suffices to prove that 
\[
(1 - \eps) \cdot \mathcal{L}_G|_{E_k} \preceq \mathcal{L}_{\widetilde{G}}|_{E_k} \preceq (1 + \eps) \cdot \mathcal{L}_G|_{E_k}
\]
for $0 \leq k \leq \lfloor n'/2 \rfloor$. To this end, let us now fix a single such $k$. Note that we can assume that $k > 0$, as for any edge $e \in E$, $(\cL_e)|_{E_0} = 0$ (thus the quadratic form on $E_0$ is preserved automatically).

Importantly, we can then observe that \begin{align}\label{eq:upperBoundEigOneTerm}
    \lambda_{\max}(X_e|_{E_k}) \leq \frac{1}{p} \cdot\lambda_{\max}(\mathcal{L}_e|_{E_k}) \leq \frac{2}{p}.
\end{align} 

At the same time, by \cref{lem:completeGraphEigenvalues}, we know that 
\[
\cL_{K_{n'}}\vert_{E_k} = k(n' + 1 - k)\cdot\Id_{E_k}.
\]
Plugging in our bound from \cref{thm:expander-Kikuchi} (along with our definition of $\lambda$), we then see that 
\[
\cL_{G}\vert_{E_k} \succeq \frac{2m}{n'^2} \cdot \lambda \cdot \cL_{K_{n'}}.
\]
Together, the preceding two claims then imply that 
\begin{align}\label{eq:lowerBoundEigExpander}
    \cL_{G}\vert_{E_k} \succeq \frac{2m}{n'^2} \cdot \lambda \cdot k(n' + 1 - k)\cdot\Id_{E_k}.
\end{align}

Now, combining \cref{eq:upperBoundEigOneTerm} and \cref{eq:lowerBoundEigExpander}, we then see that 
\[
\cL_e\vert_{E_k} \preceq \frac{1}{R} \cL_{G}\vert_{E_k},
\]
for \[
R = \frac{\frac{2m}{n'^2} \cdot \lambda \cdot k(n' + 1 - k)}{2/p} = \frac{p \cdot \frac{2m}{n'^2} \cdot \lambda \cdot k(n' + 1 - k)}{2}  = \frac{m' \cdot \frac{2m}{n'^2} \cdot \lambda \cdot k(n' + 1 - k)}{2m}
\]
\[
= \frac{\lambda \cdot k \cdot (n'+1 - k) \cdot 100 n' \log(n)}{n'^2 \cdot \lambda \eps^2} = \frac{100 \log(n) \cdot k \cdot (n' +1 -k)}{n' \cdot \eps^2}.
\]
In particular, we can observe that $n'+1 -k \geq n'/2$, and thus 
\[
R \geq \frac{50 \log(n) \cdot k}{\eps^2}.
\]
With this established, we now revisit \cref{fact:matrix-bernstein-restated}: we consider the sequence of variables $Y_e = X_e|_{E_k} \in E_k$. We then trivially have that $\E[\sum_e Y_e] = \cL_G|_{E_k}$. Thus, \cref{fact:matrix-bernstein-restated} implies that 
\[
(1 - \eps) \cdot \mathcal{L}_G|_{E_k} \preceq \mathcal{L}_{\widetilde{G}}|_{E_k} \preceq (1 + \eps) \cdot \mathcal{L}_G|_{E_k}
\]
with probability \[\geq 
1 - 2 \cdot \dim(E_k) \cdot e^{- \eps^2 \cdot R / 3} \geq 1 - 2 \cdot n'^{k+2} \cdot e^{-16 k\log(n)} \geq 1 - \frac{1}{n^{10k}}.
\]
Note that in the inequality above, we have used the value of $R$ computed above, along with \cref{prop:boundDimension} to bound $\dim(E_k) \leq (n')^{k+2}$. 

Now, we can take a union bound over all choices of $k$; in particular, the probability that for all $0 \leq k \leq \lfloor n'/2 \rfloor$ we have $(1 - \eps) \cdot \mathcal{L}_G|_{E_k} \preceq \mathcal{L}_{\widetilde{G}}|_{E_k} \preceq (1 + \eps) \cdot \mathcal{L}_G|_{E_k}$ is at least $1 - \sum_{k = 1}^{\lfloor n'/2\rfloor}\frac{1}{n^{10k}} \geq 1 - \frac{1}{n^9}$. 

Likewise, by a simple Chernoff bound, we see that with probability $1 - 1 / n^{10}$, the desired bound on $|\widetilde{G}|$ holds. Thus, we obtain our stated claim. 
\end{proof}

\section{Using Expander Decomposition to Sparsify All Kikuchi Graphs}\label{sec:expander-decomposition}

To start, we use some basic extensions of the expander decomposition results of Saranurak and Wang \cite{saranurak2019expander}. Roughly speaking, these results show that for an arbitrary graph $G$, one can decompose it into smaller subgraphs each of which is a good expander. For our purposes, we will want the additional property that the \emph{minimum-to-average} degree ratio is bounded.

\begin{theorem}\label{thm:expanderDecomp}[See, for instance, Theorem 3.4 in \cite{gupta2022online} or Theorem 2.5 in \cite{khanna2026fault}]
    Let $G = (V, E)$ be an unweighted graph with $|V| = n$ and $|E| = m$. Then, there exists an edge-disjoint decomposition of $G$ into smaller graphs $G_1, \dots G_k$ along with a universal constant $C$ such that:
    \begin{enumerate}
        \item Each vertex appears in $O(\log^2(m))$ many of the graphs $G_i: i \in [k]$.
        \item Each of the graphs $G_i$ is a $\geq \frac{1}{C \log(n)}$ conductance expander.
        \item In each graph $G_i = (V_i, E_i)$, for all $v \in V_i$, $\mathrm{deg}_{G_i}(v) \geq \frac{1}{4 C \log(n)} \cdot \frac{2|E_i|}{|V_i|}$. Equivalently, $\alpha_{G_i} \geq \frac{1}{4C\log(n)}$.
    \end{enumerate}
\end{theorem}

Intuitively, we will invoke \cref{thm:expanderDecomp} to decompose the graph $G$ into the expander components. Then, we will sparsify each of the individual expanders, and then the union of these sparsified expanders will be a sparsifier of the original graph. To do this, we first require a basic fact about the ``composability'' of sparsifiers:

\begin{proposition}\label{prop:composeSparsifier}
    Let $G = (V, E)$ be a graph, let $G_1, \dots G_k$ be a (disjoint) partition of the edges of $G$, and let $\eps \in (0, 1)$. Suppose that for each $i \in [k]$, $G'_i$ is a graph such that $(1 - \eps) \cdot \mathcal{L}_{G_i} \preceq \mathcal{L}_{G'_i} \preceq (1 + \eps) \cdot \mathcal{L}_{G_i}$. Then, for $G' = \sum_{i \in [k]} G'_i$, it must be the case that 
    \[
    (1 - \eps) \cdot \mathcal{L}_{G} \preceq \mathcal{L}_{G'} \preceq (1 + \eps) \cdot \mathcal{L}_{G}.
    \]
\end{proposition}

\begin{proof}
    Observe that $\mathcal{L}_{G'} = \sum_{i \in [k]} \cL_{G'_i}$. Thus, $\mathcal{L}_{G'} = \sum_{i \in [k]} \cL_{G'_i} \succeq \sum_{i \in [k]} (1-\eps) \cL_{G_i}$, and likewise $\mathcal{L}_{G'} = \sum_{i \in [k]} \cL_{G'_i} \preceq \sum_{i \in [k]} (1+\eps) \cL_{G_i}$.
\end{proof}

We also have the following proposition, which effectively shows that, whenever a base graph has ``unused vertices'' (in the sense that no edges are incident to these vertices), one can effectively \emph{ignore} the unused vertices. 

\begin{proposition}\label{prop:extendVertexSet}
    Let $H = (V_H, E_H)$ denote a graph on vertex set $V_H$, and let $H' = (V_H, E'_H, w)$ denote a $(1 \pm \eps)$-Kikuchi sparsifier of $H$. Then, for the graph $\hat{H} = (V, E_H)$ which lives on a larger vertex set $V_H \subseteq V$, it is the case that $\hat{H'} = (V, E'_H, w)$ is a $(1 \pm \eps)$-Kikuchi sparsifier of $H'$.
\end{proposition}

\begin{proof}
    This follows because $\cL_{\hat{H}} = \cL_{H} \otimes \Id_{V - V_H}$, and likewise $\cL_{\hat{H'}} = \cL_{H'} \otimes \Id_{V - V_H}$.
\end{proof}

\subsection{Proving the Existence of Kikuchi Sparsifiers}

With this in hand, we can prove our main theorem:

\begin{theorem}\label{thm:KikuchiSparsifierExistence}
    Let $G = (V, E)$ be an unweighted graph on $n$ vertices and $m = \mathrm{poly}(n)$ many edges, and let $\eps \in (0, 1)$. Then, there exists a re-weighted subgraph $G' = (V, E', w)$ such that 
    \[
    (1 - \eps) \cdot \mathcal{L}_G \preceq \mathcal{L}_{G'} \preceq (1 + \eps) \cdot \mathcal{L}_G
    \]
    and $|E'| \leq O(n \log^{11}(n) / \eps^2)$.
\end{theorem}

\begin{proof}
    To start, we invoke the expander decomposition statement of \cref{thm:expanderDecomp}. This yields a decomposition of our graph $G$ into $G_1, \dots G_k$ along with a universal constant $C$ such that:
    \begin{enumerate}
        \item Each vertex appears in $O(\log^2(m))$ many of the graphs $G_i: i \in [k]$.
        \item Each of the graphs $G_i$ is a $\geq \frac{1}{C \log(n)}$ conductance expander.
        \item In each graph $G_i = (V_i, E_i)$, for all $v \in V_i$, $\mathrm{deg}_{G_i}(V_i) \geq \frac{1}{4 C \log(n)} \cdot \frac{2|E_i|}{|V_i|}$. Equivalently, $\alpha_{G_i} \geq \frac{1}{4C\log(n)}$.
    \end{enumerate}

    Now, for each one of these graphs $G_i = (V_i, E_i)$, we invoke the expander sparsification statement of \cref{lem:sparsifyExpander} with $\alpha_{G_i} = \frac{1}{4 C \log(n)}$, and $\phi = \frac{1}{C \log(n)}$. If $G_i$ is a graph on $n'_i = |V_i|$ many vertices, then \cref{lem:sparsifyExpander} produces a sparsifier $G'_i = (V_i, E'_i, w)$ with 
    \[
    |E'_i| = \frac{100 n_i' \log(n)}{\lambda \eps^2} 
    \]
    many edges, for $\lambda =  \frac{\cAK \phi^4 \alpha_{G_i}^2}{128(1 + \log (n'_i/\alpha_G))^2}$. Using our values of $\alpha_{G_i}$ and $\phi$, we can see that $\lambda = \Omega(1 / \log^8(n))$. Thus, $|E'_i| = O(n'_i \log^9(n) / \eps^2)$. Note that because we are only proving an existential statement, we do not consider the failure probability in the sparsification process of \cref{lem:sparsifyExpander}; rather, \cref{lem:sparsifyExpander} suffices for proving that a sparsifier $G'_i$ \emph{exists}.

    Finally, our resulting sparsifier $G'$ is simply $\sum_{i \in [k]} G'_i$. The accuracy of the sparsifier follows from \cref{prop:composeSparsifier} and \cref{prop:extendVertexSet}. Note that the number of edges in $G'$ is bounded by \[
    \sum_{i \in [k]}|E'_i| \leq \sum_{i \in [k]}O(n'_i \log^9(n) / \eps^2) \leq O(n \log^{11}(n) / \eps^2),
    \]
    where we have used the fact that $\sum_{i \in [k]}n'_i = O(n \log^2(m)) = O(n \log^2(n))$ by the guarantees of the expander decomposition. This concludes the claim. 
\end{proof}

\subsection{An Efficient Sparsification Algorithm}

To make \cref{thm:KikuchiSparsifierExistence} algorithmic, we need an algorithmic analog of \cref{thm:expanderDecomp}. For this, we use the expander decomposition of \cite{bernstein2020fully}:

\begin{theorem}\label{thm:expanderDecompEfficient}[See, for instance, Corollary 5.6 in \cite{bernstein2020fully}]
    Let $G = (V, E)$ be an unweighted graph with $|V| = n$ and $|E| = m$. Then, there is a constant $C'$ such that for $\phi = \frac{1}{C' \log^4(m)}$, one can construct in time $O(m \log^{10}(m))$ an edge-disjoint partition of $G$ into $G_1 = (V_1, E_1), \dots G_k = (V_k, E_k)$ such that:
    \begin{enumerate}
        \item $\sum_{i \in [k]} |V_i| = O(n \log^2(m))$.
        \item Each of the graphs $G_i$ is a $\phi$ conductance expander.
        \item In each graph $G_i = (V_i, E_i)$, for all $v \in V_i$, $\mathrm{deg}_{G_i}(v) \geq \phi \cdot \frac{|E_i|}{4|V_i|}$. Equivalently, $\alpha_{G_i} \geq \frac{1}{8C'\log^4(m)}$.
    \end{enumerate}
\end{theorem}

With this theorem, we can now prove a near-linear time algorithm for producing near-linear size Kikuchi sparsifiers:

\begin{theorem}\label{thm:KikuchiSparsifierEfficient}
    Let $G = (V, E)$ be an unweighted graph on $n$ vertices and $m = \mathrm{poly}(n)$ many edges, and let $\eps \in (0, 1)$. Then, there is an $\widetilde{O}(m)$ time algorithm which computes, with probability $1 - \frac{1}{n^6}$ over the algorithm's randomness, a re-weighted subgraph $G' = (V, E', w)$ such that 
    \[
    (1 - \eps) \cdot \mathcal{L}_G \preceq \mathcal{L}_{G'} \preceq (1 + \eps) \cdot \mathcal{L}_G
    \]
    and $|E'| \leq O(n \log^{29}(n) / \eps^2)$.
\end{theorem}

\begin{proof}
    To start, we invoke the expander decomposition statement of \cref{thm:expanderDecompEfficient}. For a universal constant $C'$ , this yields a decomposition in time $O(m \log^{10}(m))$ of our graph $G$ into $G_1, \dots G_k$ such that:
    \begin{enumerate}
        \item $\sum_{i \in [k]} |V_i| = O(n \log^2(m))$.
        \item Each of the graphs $G_i$ is a $\phi = \frac{1}{C' \log^4(m)}$ conductance expander.
        \item In each graph $G_i = (V_i, E_i)$, for all $v \in V_i$, $\mathrm{deg}_{G_i}(v) \geq \phi \cdot \frac{|E_i|}{4|V_i|}$. Equivalently, $\alpha_{G_i} \geq \frac{1}{8C'\log^4(m)}$.
    \end{enumerate}

    Now, for each one of these graphs $G_i = (V_i, E_i)$, we invoke the expander sparsification statement of \cref{lem:sparsifyExpander} with $\alpha_{G_i} = \frac{1}{8C'\log^4(m)}$, and $\phi = \frac{1}{C' \log^4(m)}$. If $G_i$ is a graph on $n'_i = |V_i|$ many vertices, then by random sampling in accordance with \cref{lem:sparsifyExpander}, we produce a sparsifier $G'_i = (V_i, E'_i, w)$ with 
    \[
    |E'_i| = \frac{100 n_i' \log(n)}{\lambda \eps^2} 
    \]
    many edges, for $\lambda =  \frac{\cAK \phi^4 \alpha_{G_i}^2}{128(1 + \log (n'_i/\alpha_G))^2}$. Using our values of $\alpha_{G_i}$ and $\phi$, we can bound $\lambda = \Omega(1 / \log^{26}(m)) = \Omega(1 /\log^{26}(n))$. Thus, $|E'_i| = O(n'_i \log^{27}(n) / \eps^2)$. Moreover, the sparsification statement of \cref{lem:sparsifyExpander} holds with probability $\geq 1 - \frac{1}{n^{8}}$.

    Finally, our resulting sparsifier $G'$ is simply $\sum_{i \in [k]} G'_i$. The sparsifier accuracy follows from \cref{prop:composeSparsifier} and \cref{prop:extendVertexSet}. Note that the number of edges in $G'$ is bounded by \[
    \sum_{i \in [k]}|E'_i| \leq \sum_{i \in [k]}O(n'_i \log^{27}(n) / \eps^2) \leq O(n \log^{29}(n) / \eps^2),
    \]
    where we have used the fact that $\sum_{i \in [k]}n'_i = O(n \log^2(m)) = O(n \log^2(n))$ by the guarantees of the expander decomposition. Finally, we can take a union bound over all $k \leq O(n \log^2(m))$ many of the expander components; thus, with probability $\geq 1 - \frac{1}{n^6}$, the resulting graph $G'$ is indeed a $(1 \pm \eps)$ Kikuchi sparsifier of $G$ and satisfies our space bound.
\end{proof}

\section{Extension to Weighted Graphs}\label{sec:weighted}

So far, we have shown that \cref{thm:main} holds for \emph{unweighted} graphs. We now extend the sparsifer to cover weighted graphs as well.

\subsection{Limited Aspect Ratio}

As a first step, we show that we can always get a variant of \cref{thm:main} if we are willing to pay a cost proportional to $\log \frac{\max w}{\min w}$.

\begin{theorem}\label{thm:WeightedSparsifierExistenceSloppy}
    Let $G = (V, E, w)$ be a weighted graph on $n$ vertices, and let $\eps \in (0,1)$. Then, there exists a re-weighted subgraph $G' = (V, E', w')$ such that 
    \[
    (1 - \eps) \cdot \mathcal{L}_G \preceq \mathcal{L}_{G'} \preceq (1 + \eps) \cdot \mathcal{L}_G
    \]
    and $|E'| \leq O(n \log^{9}(n)\log^2\left(\frac{n \max w}{\min w}\right) / \eps^2)$.
\end{theorem}

\begin{proof}
If $\eps < 1/n$, then we can return $G$ itself and only have $n^2 = O(n/\eps^2)$ edges. Thus, assume $\eps \ge 1/n$. Let $\widetilde{G} = (V, \widetilde{E})$ be an unweighted graph such that for every $e \in E$, there are
\[
    \widetilde{w}(e) = \left\lfloor \frac{4}{\eps} \cdot \frac{w(e)}{\min w}\right\rfloor
\]
copies of $e$ in $\widetilde{E}$. Observe that
\[
m := |\widetilde{E}| \le \sum_{e \in E} \left\lfloor \frac{4}{\eps} \cdot \frac{w(e)}{\min w}\right\rfloor \le n^2 \cdot \frac{3}{\eps} \cdot  \frac{\max w}{\min w} \le 3n^3 \frac{\max w}{\min w}.
\]
Upon careful inspection of the proof of \cref{thm:KikuchiSparsifierExistence}, one can build a sparsifier $\widetilde{G}' = (V, \widetilde{E}', \widetilde{w}')$ of $\widetilde{G}$ with 
\[
    |\widetilde{E}'| \le O(n \log^9(n) \log^2(m) / \eps^2) = O(n \log^{9}(n)\log^2(n \max w / \min w) / \eps^2)
\]
and
\[
  \left(1 - \frac{\eps}{4}\right) \cL_{\widetilde{G}} \preceq  \cL_{\widetilde{G}'} \preceq \left(1 + \frac{\eps}{4}\right) \cL_{\widetilde{G}}.
\]
We now build $G' = (V, E', w')$ by setting $E' = \widetilde{E}'$ and setting for all $e \in E'$,
\[
    w'(e) = \frac{\eps}{4} \cdot \widetilde{w}'(e) \cdot \min w.
\]
Thus,
\[
\frac{\eps}{4} \cdot \min w \cdot \left(1 - \frac{\eps}{4}\right) \cL_{\widetilde{G}} \preceq \frac{\eps}{4} \cdot \min w \cdot \cL_{\widetilde{G}'} = \cL_{G'} \preceq \frac{\eps}{4} \cdot \min w \cdot \left(1 + \frac{\eps}{4}\right) \cL_{\widetilde{G}}.
\]
Note that for all $e \in E$, we have that
\[
    \frac{\eps}{4} \cdot \min w \cdot \widetilde{w}(e) \le \frac{\eps}{4} \cdot \min w \cdot \frac{4}{\eps} \cdot \frac{w(e)}{\min w} = w(e),
\]
and since $\eps < 1$,
\[
    \frac{\eps}{4} \cdot \min w \cdot  \widetilde{w}(e) \ge \frac{\eps}{4} \cdot \min w \cdot \left[\frac{4}{\eps} \cdot \frac{w(e)}{\min w} - 1\right] \ge \frac{\eps}{4} \cdot \min w \cdot \frac{3}{\eps} \cdot \frac{w(e)}{\min w} \ge \left(1 - \frac{\eps}{4}\right) w(e).
\]
Therefore,
\[
(1 - \eps) \cL_G \preceq \left(1 - \frac{\eps}{4}\right)^2 \cL_{G} \preceq \frac{\eps}{4} \cdot \min w \cdot \left(1 - \frac{\eps}{4}\right) \cL_{\widetilde{G}} \preceq \cL_{G'}\]
and
\[
\cL_{G'} \preceq \frac{\eps}{4} \cdot \min w \cdot \left(1 + \frac{\eps}{4}\right) \cL_{\widetilde{G}} \preceq \left(1 + \frac{\eps}{4}\right) \cL_{G} \preceq (1 + \eps) \cL_G,
\]
as desired.
\end{proof}

\subsection{An Effective Resistance Inequality}

Given a weighted graph $G = (V, E, w)$, we let $R_{G}(u, v)$ denote the effective resistance between vertices $u$ and $v$ in an electrical network with conductances indicated by $w : E \to \R_{> 0}$. More formally, for any edge $e = \{u, v\}$ we define the effective resistance $R_G(e)$ to be
\[R_{G}(e) = \sup_{\substack{x \in \R^V\setminus\{0\}\\\langle x, \1\rangle = 0}} \frac{(x_u - x_v)^2}{\sum_{f = \{a, b\}\in E(G)}w_f\cdot(x_a - x_b)^2}.\] 

If $u$ and $v$ are in different connected components, then $R_{G}(u, v) = \infty$.  We recall the following basic facts about effective resistance.

\begin{proposition}\label{prop:eff-resistance}
Let $G = (V, E, w)$ be a weighted graph. The following facts hold.
\begin{enumerate}
\item For any edge $e \in E$, $R_{G}(e) \le \frac{1}{w_e}$.\label{item:eff1}
\item \cite[Lemma~12.9.1]{Spielman25} $R_{G} : V \times V \to \R_{\ge 0} \cup \{\infty\}$ forms a metric. \label{item:eff2}
\end{enumerate}
\end{proposition}

\begin{proof}
Since \cref{item:eff2} is proved in \cite[Lemma~12.9.1]{Spielman25}, we focus on proving \cref{item:eff1}, if $e = \{u, v\} \in E$ then
\[
    R_{G}(e) = \sup_{\substack{x \in \R^V\setminus\{0\}\\\langle x, \1\rangle = 0}} \frac{(x_u - x_v)^2}{\sum_{f = \{a, b\}\in E(G)}w_f\cdot(x_a - x_b)^2} \le \sup_{\substack{x \in \R^V\setminus\{0\}\\\langle x, \1\rangle = 0}}  \frac{(x_u - x_v)^2}{w_e (x_u - x_v)^2} = \frac{1}{w_e},
\]
as desired.
\end{proof}

To relate effective resistances with the spectrum of the Quantum Cut Hamiltonian $\cL_{G}$, we make use of the following theorem due to Chen~\cite{Chen17}.

\begin{theorem}[Theorem 1.1~\cite{Chen17}]\label{thm:chen}
Let $G = (V, E, w)$ be a weighted graph. For any $u, v \in V$ (for which it is not necessarily the case that $\{u, v\} \in E$), we have that
\[
    \cL_{G} R_{G}(u, v) \succeq \cL_{u, v}.
\]
\end{theorem}

\begin{remark}
Although Chen is not studying quantum cut, he built a linear operator identical to $\cL_{G}$ by studying a Markov chain on the hypercube $\{0,1\}^V$, where with probability proportional to $w_{e}$, the bits correspond to the edge $e$ get swapped. Similar to the inequality of Alon-Kozma~\cite{AlonKozma20}, Theorem~\ref{thm:chen} crucially uses the octopus inequality of \cite{CLR10} in its proof.
\end{remark}

\begin{remark}
Note that on its own \cref{thm:chen} is not sufficient to build Kikuchi sparsifiers. This is due to the fact that our sparsification results rely on the relative importance of an edge $u,v$ \emph{decreasing} in the higher harmonic levels. As an example, in the complete graph $K_n$, \cref{thm:chen} roughly says that $\frac{1}{n-1} \cL_{K_n}  \succeq \cL_{u, v}$. While this on its own is tight, our sparsification argument roughly shows that, if we restrict to the orthogonal complement of a space of dimension $\binom{n}{i}$ (for $i \leq n/2$), we can in fact show the much stronger inequality that $\frac{1}{i \cdot n} \cL_{K_n}|_{E_i}  \succeq \cL_{u, v}|_{E_i}$.

Interestingly, this stronger inequality (properly normalized) fails to hold for general graphs. Indeed, let $G = ([n], E)$ be a tree on $n$ vertices (with the weight of every edge being $1$), and fix any edge $e\in E(G)$. Note that the effective resistance of $e$ is exactly $1$, and thus \cref{thm:chen} says that $\cL_G\succeq\cL_e$. We now claim that for any $\delta > 0$, $(1 - \delta)\cL_G\vert_{E_i}\not\succeq\cL_e\vert_{E_i}$ for \emph{every} $1\leq i\leq n - 1$, i.e. the ``relative importance/effective resistance'' of $e$ does not decay at all even as we go to ``higher'' eigenspaces! Indeed, first note that since $G$ is a tree, deleting $e$ disconnects $G$ into exactly two connected components $C_u, C_v\subset[n]$, where $u\in C_u, v\in C_v$. Now consider the function $\psi\in V_i = \C^{\binom{[n]}{i}}$ defined as $\psi(S) := \1(|S\cap C_u|\text{ is even})$. It is easy to see that $\langle\psi, (\cL_e\vert_{E_i})\psi\rangle = \binom{n - 2}{i - 1} > 0$, while $\langle\psi, (\cL_{e'}\vert_{E_i})\psi\rangle = 0$ for all $e'\in E(G)\setminus\{e\}$, which proves our claim since $\cL_G\vert_{E_i} = \sum_{e\in E(G)}\cL_e\vert_{E_i}$. 

Intuitively, the idea is that for any graph $G$, if $e\in E(G)$ is an ``extremely important'' edge, then the relative importance of $e$ might not decay in the higher eigenspaces. However, one implication of expander decompositions is that \emph{most} edges in a graph (with $\gg n\log^{O(1)}(n)$ edges, say) can \emph{not} be very important. For those edges, we indeed obtain some decay in the higher eigenspaces, and thus even though a decay of the sort $(1/i)\cdot R_e\cL_G\vert_{E_i}\succeq\cL_e\vert_{E_i}$ might not hold for every edge $e$, it holds on ``average'' (upto some log factors), and that is enough for us to obtain a sparsifier.
\end{remark}

\subsection{An Edge-moving Lemma}

Given weighted graphs $G = (V, E, w)$ and $H = (V, E', w')$ on the same vertex set $V$, we define $G \uplus H = (V, E'', w'')$ to be the weighted graph on the edge set $E'' = E \cup E'$ such that for all $e \in E''$, we have that
\[
w''(e) = \begin{cases}
w(e) + w'(e) & e \in E \cap E'\\
w(e) & e \in E \setminus E'\\
w'(e) & e \in E' \setminus E.
\end{cases}
\]
From the definition, it follows that $\cL_{G \uplus H} = \cL_{G} + \cL_{H}$. Consider two graphs $G = (V, E, w)$ and $H = (V, E', w')$ on the same vertex set but their edge weights are widely different. For example, in \cref{subsec:building-weighted}, we consider the possibility that $\min w \ge n^{20} \max w'$, where $n = |V|$. We would like to say that if an edge $e \in E'$ spans two connected components of $G$, then $e$ can be ``moved'' to different vertices between the same two connected components with no consequence. We formalize this as follows

\begin{definition}\label{def:equivalent}
Let $G = (V, E, w)$ be a weighted graph with connected components $V_1, \hdots, V_r$. Let $H = (V, E', w')$ and $H' = (V, E'', w'')$ be weighted graphs on the same vertex set. We say that $H$ and $H'$ are \emph{equivalent} with respect to $G$ if for all $1 \le i < j \le r$, we have that
\[
    \sum_{e \in E' \cap (V_i \times V_j)} w'(e) = \sum_{e \in E'' \cap (V_i \times V_j)} w''(e).
\]
\end{definition}

Observe that the equivalence of $H$ and $H'$ does not take into account edges which are within any connected component $V_i$ of $G$. The goal of this section is to prove the following ``Edge-moving Lemma.''

\begin{lemma}[Edge-moving Lemma]\label{lem:aggregate}
Let $G = (V, E, w)$ be a weighted graph, and let $H = (V, E', w')$ and $H' = (V, E'', w'')$ be weighted graphs which are equivalent with respect to $G$. Let $\lambda := \frac{\min w}{\max w' + \max w''}$. Then, for any $\eta \in (0, 1/2]$, we have that
\begin{align}
    \left(1 - \frac{n^3}{\eta \lambda} - \eta\right) \cL_{G \uplus H} \preceq \cL_{G \uplus H'} \preceq \left(1 + \frac{n^3}{\eta \lambda} + \eta\right) \cL_{G \uplus H}.\label{eq:aggregate}
\end{align}
\end{lemma}

As first step toward proving \cref{lem:aggregate}, we show that edges in $H$ which are within a single connected component of $G$ can be effectively ignored.

\begin{lemma}\label{lem:remove-within}
Let $G = (V, E, w)$ be a weighted graph with connected components $V_1, \hdots, V_r$. Pick $1 \le i \le r$ as well as any distinct $u, v \in V_i$. We have that
\[
    \frac{n}{\min w}\cL_{G} \succeq \cL_{u, v}.
\]
\end{lemma}

\begin{proof}
By \cref{thm:chen}, it suffices to prove that $\frac{n}{\min w} \ge R_{G}(u, v)$. To see why, by \cref{prop:eff-resistance}, \cref{item:eff1}  we know that $R_G(e) \le \frac{1}{\min w}$ for all $e \in E$. Since $u$ and $v$ are in the same connected component, there is a chain of at most $n-1$ edges $e_1, \hdots, e_\ell$ connecting $u$ and $v$. By \cref{prop:eff-resistance}, \cref{item:eff2} we have that \[R_G(u, v) \le R_G(e_1) + \cdots + R_G(e_\ell) \le \frac{\ell}{\min w} \le \frac{n}{\min w},\] as desired.
\end{proof}

Second, we show that we can move one edge between connected components without affecting the sparsifier too much.

\begin{lemma}\label{lem:move-between}
Let $G = (V, E, w)$ be a weighted graph with connected components $V_1, \hdots, V_r$. Pick $1 \le i < j \le r$ as well as $u_i, v_i \in V_i$ and $u_j, v_j \in V_j$. For any $\delta > 0$, we have that
\[
    \frac{(1+\delta)n}{\delta \min w}\cL_{G} + (1+\delta)\cL_{u_i, u_j} \succeq \cL_{v_i, v_j}.
\]
\end{lemma}

\begin{proof}
Define $G' = (V, E', w')$ to be a weighted graph where $E' = E \cup \{(u_i,u_j)\}$ and for all $e \in E'$,
\[
    w'(e) = \begin{cases}
    \frac{n}{\delta \min w}w(e) & e \in E\\
    1 & e = \{u_i, u_j\}.
    \end{cases}
\]
By definition, we have that $\cL_{G'} = \frac{2n}{\delta \min w}\cL_G + \cL_{u_i, u_j}$. By \cref{thm:chen}, we thus have that
\[
    R_{G'}(v_i, v_j) \left[\frac{n}{\delta \min w}\cL_{G} + \cL_{u_i, u_j}\right] \succeq \cL_{v_i, v_j}.
\]
Thus, it suffices to prove that $R_{G'}(v_i, v_j) \le 1 + \delta$.  Similar to the proof of \cref{lem:remove-within}, we make use of \cref{prop:eff-resistance} to argue this fact. In particular, by \cref{prop:eff-resistance}, \cref{item:eff1} we have for all $e \in E$ that
\[
R_{G'}(e) \le \frac{1}{w'(e)} = \frac{1}{\frac{2n}{\delta \min w} w(e)} \le \frac{1}{\frac{2n}{\delta \min w} \min w} = \frac{\delta}{n}. 
\]
Now, consider a shortest path from $v_i$ to $v_j$ in $G'$. This path will use at most $n-2$ edges $e \in E$ as well as the bridge edge $\{u_i, u_j\}$. By \cref{prop:eff-resistance}, \cref{item:eff2} we can thus deduce that
\[
    R_{G'}(v_i, v_j) \le R_{G'}(u_i, u_j) + (n-2) \cdot \frac{\delta}{n} \le 1 + \delta, 
\]
where the last step uses that $R_{G'}(u_i, u_j) \le 1$ by \cref{prop:eff-resistance}, \cref{item:eff1}.\end{proof}

We are now equipped to prove \cref{lem:aggregate}.

\begin{proof}[Proof of \cref{lem:aggregate}.]
If $n = 1$, the inequality is trivial, so we assume that $n \ge 2$. First, write $H = H_0 \uplus H_1$ where $H_0 = (V, E'_0, w'_0)$ and $H_1 = (V, E'_1, w'_1)$, where $E'_0 \subseteq E'$ is the set of edges of $H$ within some component of $G$ and $E'_1 \subseteq E'$ is the set of edges which cross two components of $G$. We define $H' = H'_0 \uplus H'_1$ analogously. Observe that by \cref{lem:remove-within}, we have that
\begin{align}
    \cL_{H'_0} = \sum_{e \in E''_0} w''(e) \cL_{e} \preceq \sum_{e \in E''_0} \frac{n w''(e)}{\min w} \cL_{G} \preceq \frac{n^3 \max w''}{\min w} \cL_G \preceq \frac{n^3}{\lambda} \cL_{G}.\label{eq:H'_0}
\end{align}

We now seek to relate $\cL_{H_1}$ with $\cL_{H'_1}$. To do this, we first form what we call a \emph{movement plan} between $H_1$ and $H'_1$. More precisely, we say that $\cM = \{(\lambda_i, e_i, f_i) : i \in [L]\}$ is a movement plan between $H_1$ and $H'_1$ if for all $i \in [L]$ we have that $e_i \in E'_1$, $f_i \in E''_1$ and that $e_i$ and $f_i$ bridge the same pair of connected components of $G$. We further require that
\[
    \cL_{H_1} = \sum_{i = 1}^L \lambda_i \cL_{e_i} \text{\ \ \ and \ \ \ } \cL_{H'_1} = \sum_{i = 1}^L \lambda_i \cL_{f_i}.
\]
By \cref{def:equivalent}, such a movement plan always exists and may take it to have size $L \le |E'_1| + |E''_1| \le 2\binom{n}{2} \le n^2$ by repeatedly picking arbitrary edges $e_i \in E'_1$ and $f_i \in E''_1$ whose movements have not been fully assigned and setting $\lambda_i$ to be as large as possible.

We now make crucial use of \cref{lem:move-between}. Set $\delta = \frac{n^3}{\eta \lambda}$. For any $i \in [L]$, we have the following inequality
\begin{align*}
\lambda_i\frac{(1+\delta)n}{\delta \min w}\cL_{G} + (1+\delta)\lambda_i\cL_{e_i} &\succeq \lambda_i\cL_{f_i}.
\end{align*}
Summing this inequality over all $i \in [L]$, we have that
\begin{align*}
\left(\sum_{i=1}^L \lambda_i\right)\frac{(1+\delta)n}{\delta \min w}\cL_{G} + (1+\delta)\cL_{H_i} &\succeq \cL_{H'_1}.
\end{align*}
Observe that $\sum_{i=1}^L \lambda_i = \sum_{e \in E'_1} w'(e) \le n^2 \max w'$. Thus, we actually have that
\[
    \left(\sum_{i=1}^L \lambda_i\right)\frac{(1+\delta)n}{\delta \min w} \le (1 + \delta) \frac{n^3 \max w'}{\delta \min w} = (1 + \delta) \frac{\eta \lambda \max w'}{\min w}\le (1+\delta)\eta.
\]
That is,
\begin{align}
    \left(\frac{n^3}{\lambda} + \eta\right)\cL_{G} + \left(1 + \frac{n^3}{\eta \lambda}\right)\cL_{H_1} \ge \cL_{H'_1}.\label{eq:H'_1}
\end{align}

We now seek to prove the RHS of \cref{eq:aggregate}. Observe that
\begin{align*}
    &\left(1 + \frac{n^3}{\eta\lambda} + \eta\right) \cL_{G \uplus H}\\
    &\succeq \cL_{G} + \left(\frac{n^3}{\lambda}\right)\left(\frac{1}{\eta} - 1\right) \cL_G + \left(\frac{n^3}{\lambda} + \eta\right) \cL_{G} + \left(1 + \frac{n^3}{\eta\lambda}\right)\cL_{H_1}\\
    &\succeq \cL_{G} + \left(\frac{n^3}{\lambda}\right)\cL_G + \left[\left(\frac{n^3}{\lambda} + \eta\right) \cL_{G} + \left(1 + \frac{n^3}{\eta\lambda}\right)\cL_{H_1}\right] & (\eta \le 1/2)\\
    &\succeq \cL_{G} + \cL_{H'_0} + \cL_{H'_1} & (\text{\cref{eq:H'_0}, \cref{eq:H'_1}})\\
    &= \cL_{G \uplus H'}.
\end{align*}
By symmetry, we know that $\left(1 + \frac{n^3}{\eta\lambda} + \eta\right) \cL_{G \uplus H'} \succeq \cL_{G \uplus H}$, from which we can deduce the LHS of \cref{eq:aggregate} by the identity that $\frac{1}{1+x} \ge 1 - x$ for all $x \ge 0$. Thus, \cref{eq:aggregate} holds.\end{proof}

\subsection{Building the Sparsifier}\label{subsec:building-weighted}

We now prove our main weighted sparsifier theorem.

\begin{theorem}\label{thm:WeightedSparsifierExistence}
    Let $G = (V, E, w)$ be a weighted graph on $n$ vertices, and let $\eps \in (0, 1)$. Then, there exists a re-weighted subgraph $G' = (V, E', w')$ such that 
    \[
    (1 - \eps) \cdot \mathcal{L}_G \preceq \mathcal{L}_{G'} \preceq (1 + \eps) \cdot \mathcal{L}_G
    \]
    and $|E'| \leq O(n \log^{11}(n) / \eps^2)$.
\end{theorem}

\begin{proof}
We may assume that $n$ is at least a sufficiently large constant. Without loss of generality, we may normalize the weights of $G$ so that $\max w = 1$. We also may assume without loss of generality that $\eps \ge 1/n$. Set $\kappa = n^{-20}$, for integers $i \ge 0$, we define the $i$th weight bucket to be $B_i = [\kappa^i, \kappa^{i+1}]$. Likewise, let $G_i$ be $G$ restricted to the weights in the $i$th bucket. Despite there being infinitely many buckets, at most $\binom{n}{2}$ have an edge. Let $I$ be the set of indices $i \in I$ for which $G_i$ is nonempty. Now define
\begin{align*}
G_{\mathrm{even}} = \biguplus_{\substack{i \in I\\i\text{ even}}} G_i \text{ \ \ \  and \ \ \ }
G_{\mathrm{odd}} = \biguplus_{\substack{i \in I\\i\text{ odd}}} G_i.
\end{align*}
Clearly to sparsify $\cL_{G}$ it suffices to sparsify each of $\cL_{G_{\mathrm{even}}}$ and $\cL_{G_{\mathrm{odd}}}$. We now demonstrate how to sparsify $\cL_{G_{\mathrm{even}}}$, with the sparsification of $\cL_{G_{\mathrm{odd}}}$ following by near-identical methods.

Given an even integer $t \ge 0$, define
\[
    G_{\mathrm{even}, \le t} = \biguplus_{i \in I \cap \{0, 2, 4, \hdots, t\}} G_i.
\]

Define $C_t$ and $M_t$ to be the number of connected components and edges which $G_{\mathrm{even}, \le t}$ has. We seek to prove by induction the following statement.

\begin{proposition}\label{prop:induct}
For all $t$, we have that $G_{\mathrm{even}, \le t}$ has an $\frac{\eps}{2} + \frac{M_t}{n^5}$-sparsifier $G_{\mathrm{even}, \le t}'$  with at most $f(n - C_t)$ edges, where $f(n) = \Theta(n \log^{11}(n) / \eps^2)$ is a subadditive function.
\end{proposition}

As a base case of $t = 0$, assume $G_0$ has components $V_1, \hdots, V_r$ where $r = C_0$. We can apply \cref{thm:WeightedSparsifierExistenceSloppy} on each component to get an $\eps/2$-sparsifier with at most $f(|V_i| - 1)$ edges (using the fact that $\log(1/\kappa) = 20 \log n$). When these components are aggregated, the total sparsifier size is
\[
    \sum_{i=1}^{r} f(|V_i|-1) \le f(n - r) = f(n - C_0),
\]
as desired.

Now consider $t \ge 2$ for which we assume by the inductive hypothesis that \cref{prop:induct} holds. If $G_{t+2}$ is empty, then the inductive step is trivial. Thus, assume that $G_{t+2}$ has at least one new edge. Let $V_1, \hdots, V_r$ (with $r = C_t$) be the connected components of $G_{\mathrm{even}, \le t}$, and pick arbitrary representatives $v_1 \in V_1, \hdots, v_r \in V_r$. Define a new graph $H_{t+2} = (V, E_H, w_H)$ such that for all $1 \le i < j \le n$ where is an edge between $\{v_i, v_j\} \in E_H$ with $w_H(v_i, v_j)$ being the total sum of weights of edges between $V_i$ and $V_j$ in $G_{t+2}$ (unless $w_H(v_i, v_j) = 0$ in which case we do not include $\{v_i,  v_j\}$ in $E_H$). One can verify that $H_{t+2}$ has exactly $C_{t+2}$ connected components.

Let $H'_{t+2} = (V, E'_H, w'_H)$ be an $(\eps/2)$-sparsifier of $H_{t+2}$ with $f(C_t - C_{t+2})$ edges. For every $\{v_i, v_j\} \in E'_H$, let $\{u^{ij}_i, u^{ij}_j\} \in E$ be an edge of $G_{t+2}$ which goes between $V_i$ and $V_j$.

Let $G'_{t+2} = (V, E'_G, w'_G)$ be a weighted graph where $E'_G = \{\{u^{ij}_i, u^{ij}_j\}  \mid \{v_i, v_j\} \in E'_H\}$ and $w'_G(u^{ij}_i, u^{ij}_j) = w'_H(v_i, v_j)$. We claim that $G'_{\mathrm{even}, \le t+2}$ can be chosen to be $G'_{\mathrm{even}, \le t} \uplus G'_{t+2}$. First, the edge count works out because $G'_{t+2}$ has the same number of edges as $H'_{t+2}$ which is $f(C_t - C_{t+2})$ so $G'_{\mathrm{even}, \le t+2}$ has at most
\[
    f(n - C_t) + f(C_t - C_{t+2}) \le f(n - C_{t+2})
\]
edges, as desired. To check the accuracy of this sparsifier, we invoke \cref{lem:aggregate} multiple times with $\eta = 1/n^6$ and $\lambda \ge \kappa^{t+1}/(2n^2\kappa^{t+2}) \ge n^{17}$. First, since $G_{t+2}$ and $H_{t+2}$ are equivalent, we may invoke the lemma to get
\[
        \left(1 - \frac{2}{n^6}\right) \cL_{G_{\mathrm{even}, \le t} \uplus G_{t+2}} \preceq \cL_{G_{\mathrm{even}, \le t} \uplus H_{t+2}} \preceq \left(1 + \frac{2}{n^6}\right) \cL_{G_{\mathrm{even}, \le t} \uplus G_{t+2}}.
\]
Likewise, since $G'_{t+2}$ and $H'_{t+2}$ are equivalent, we may invoke the lemma again\footnote{Here we observe by the construction of the unweighted sparsifier in \cref{sec:expander-decomposition}, the minimum edge weight of $G'_{\mathrm{even}, \le t}$ is at least the minimum edge weight of $G_{\mathrm{even}, t}$.} with respect to the graph $G'_{\mathrm{even}, \le t}$ using $\eta = 1/n^6$ and\footnote{The choice of $\lambda$ uses the fact that no sparsifier of $G_{t+2}$ or $H_{t+2}$ can have any edge weight exceed $(1+\eps)n^2\kappa^{t+2} \le 2n^2\kappa^{t+2}$.} $\lambda \ge \kappa^{t+1}/(4n^2\kappa^{t+2}) \ge n^{17}$ to get that
\[
    \left(1 - \frac{2}{n^6}\right) \cL_{G'_{\mathrm{even}, \le t} \uplus H'_{t+2}} \preceq \cL_{G'_{\mathrm{even}, \le t} \uplus G'_{t+2}} \preceq \left(1 + \frac{2}{n^6}\right) \cL_{G'_{\mathrm{even}, \le t} \uplus H'_{t+2}}.
\]

Since $H'_{t+2}$ is an $(\eps/2)$-sparsifier of $H_{t+2}$ and $G'_{\mathrm{even}, \le t}$ is an $\frac{\eps}{2} + \frac{M_t}{n^5}$-sparsifier $G_{\mathrm{even}, \le t}$ we also have that
\[
    \left(1 - \frac{\eps}{2} - \frac{M_t}{n^5}\right)\cL_{G_{\mathrm{even}, \le t} \uplus G_{t+2}} \le \cL_{G'_{\mathrm{even}, \le t} \uplus H'_{t+2}} \preceq \left(1 + \frac{\eps}{2} + \frac{M_t}{n^5}\right)\cL_{G_{\mathrm{even}, \le t} \uplus G_{t+2}}.
\]
Combining these three equations, we get that
\[
    \left(1 - \frac{\eps}{2} - \frac{M_t}{n^5}\right)\left(1 - \frac{2}{n^6}\right)^2\cL_{G_{\mathrm{even}, \le t} \uplus G_{t+2}} \preceq \cL_{G'_{\mathrm{even}, \le t} \uplus G'_{t+2}} \preceq \left(1 + \frac{\eps}{2} + \frac{M_t}{n^5}\right)\left(1 + \frac{2}{n^6}\right)^2\cL_{G_{\mathrm{even}, \le t} \uplus G_{t+2}} 
\]
Since $M_t + 1 \le M_{t+2}$ and $n$ is sufficiently large, we get that \[
\left(1 - \frac{\eps}{2} - \frac{M_{t+2}}{n^5}\right)\cL_{G_{\mathrm{even}, \le t} \uplus G_{t+2}} \preceq \cL_{G'_{\mathrm{even}, \le t} \uplus G'_{t+2}} \preceq \left(1 + \frac{\eps}{2} + \frac{M_{t+2}}{n^5}\right)\cL_{G_{\mathrm{even}, \le t} \uplus G_{t+2}},
\]
completing the proof of the inductive hypothesis \cref{prop:induct}. From this, \cref{thm:WeightedSparsifierExistence} immediately follows.
\end{proof}

\section{Conclusion}

In this paper, we showed how to combine a harmonic decomposition of hypercube slices, a powerful operator inequality of Alon and Kozma, and an expander decomposition framework to efficiently compute near-optimal sparsifiers of all Quantum Cut Hamiltonians. We conclude with a few potential directions for future research.

\paragraph{Tightening \cref{thm:main}.}
Our current method of proving \cref{thm:main} necessarily pays a number of $\log n$ factors in the size of the sparsifier. On the other hand, in the classical setting, results like \cite{BatsonSS14} show that $O(n/\eps^2)$-sized sparsifiers are possible for spectral sparsifiers. Are $O_{\eps}(n)$-sized quantum cut sparsifiers possible as well?

\paragraph{Sparsifiers for Broader Families of Local Hamiltonians.} As pursued in \cite{BasuBP26}, sparisfying quantum cut is one example of a large variety of local Hamiltonians which one could hope to sparsify. Our \cref{thm:main} not only resolves Question 10.2 of~\cite{BasuBP26}, but our methods also give partial insight into their Question 10.1 on relating the non-redundancy of systems of local Hamiltonians to their ultimate sparsifiability. As a first goal, we believe that with tools like the Alon-Kozma inequality~\cite{AlonKozma20}, it should be feasible to classify the sparsifiability of all 2-local Hamiltonian systems.

\paragraph{Further Applications of the Octopus Inequality.} The key technical tool toward proving the Alon-Kozma inequality was the octopus inequality designed by Caputo, Liggett, and Richthammer~\cite{CLR10}. Given our paper is the first to apply this octopus inequality to sparsification, we hope that our work encourages the study of more applications of the octopus inequality to various problems in quantum physics and theoretical computer science.

Of note, if we wish to apply these ideas to $k$-local Hamiltonian systems with $k \geq 3$, we likely need to understand \emph{hypergraph} generalizations of the octopus inequality. Recent works by Alon, Kozma, and Puder~\cite{AlonKP25,AlonP26} have started laying the groundwork for understanding the octopus inequality in these hypergraph settings.

\bibliographystyle{alphaurl}
\bibliography{quantum}

\appendix

\end{document}